\UseRawInputEncoding

\documentclass[%
 reprint,
 superscriptaddress,
 amsmath,amssymb,
 aps,
 pra,
 floatfix,
]{revtex4-2}

\usepackage{graphicx}
\usepackage{dcolumn}
\usepackage{bm}
\usepackage{xcolor}
\usepackage{hyperref}
\definecolor{darkred}  {rgb}{0.5,0,0}
\definecolor{darkblue} {rgb}{0,0,0.5}
\definecolor{darkgreen}{rgb}{0,0.5,0}
\hypersetup{
  colorlinks = true,
  urlcolor  = blue,         
  linkcolor = darkblue,     
  citecolor = darkgreen,    
  filecolor = darkred       
}

\usepackage{bm,amsmath,amssymb,physics}
\usepackage{dsfont,enumerate,comment}
\usepackage{array}

\usepackage{amsthm}
\newtheorem{theorem}{Theorem}
\newtheorem{lemma}[theorem]{Lemma}
\newtheorem{definition}[theorem]{Definition}

\usepackage[capitalise]{cleveref}

\newcommand{\del}{\partial}

\newcommand{\tx}[1]{\textrm{#1}}

\newcommand{\lrp}[1]{\left(#1\right)}
\newcommand{\lrb}[1]{\left[#1\right]}
\newcommand{\lrcb}[1]{\left\{ #1 \right\} }
\newcommand{\ii}{\mathrm{i}}
\newcommand{\R}{\mathbb{R}}
\newcommand{\floor}[1]{\lfloor #1 \rfloor}

\newcommand{\C}{\mathbb{C}}
\newcommand{\Hammy}{\mathcal{H}}

\newcommand{\ident}{\mathds{1}}

\usepackage{xcolor}
\definecolor{ginger}{rgb}{0.69, 0.4, 0.0}
\definecolor{dgreen}{RGB}{44,162,95}

\newcommand{\mH}{{\mathcal H}}
\newcommand{\mD}{{\mathcal D}}
\newcommand{\mS}{{\mathcal S}}
\newcommand{\mR}{{\mathcal R}}

\newcommand{\specialcell}[2][c]{%
    \begin{tabular}[#1]{@{}c@{}}#2\end{tabular}}

\begin{document}

\preprint{APS/123-QED}

\title{Quantum Error Corrected Non-Markovian Metrology}

\author{Zachary Mann}
\email{zmann@uwaterloo.ca}
\affiliation{%
Institute for Quantum Computing, University of Waterloo, Waterloo, Ontario N2L 3G1, Canada
}%
\affiliation{%
Department of Physics and Astronomy, University of Waterloo, Waterloo, Ontario N2L 3G1, Canada
}%
 
\author{Ningping Cao}%
\affiliation{%
 Perimeter Institute for Theoretical Physics, Waterloo, Ontario N2L 2Y5, Canada
}%
\affiliation{%
 Institute for Quantum Computing, University of Waterloo, Waterloo, Ontario N2L 3G1, Canada
}%
\affiliation{
Digital Technologies National Research Council Collaboration Center,
East Campus 5, Waterloo, Ontario, N2L 3G1, Canada
}

\author{Raymond Laflamme}%
\affiliation{%
 Institute for Quantum Computing, University of Waterloo, Waterloo, Ontario N2L 3G1, Canada
}%
\affiliation{%
Department of Physics and Astronomy, University of Waterloo, Waterloo, Ontario N2L 3G1, Canada
}%
\affiliation{%
Perimeter Institute for Theoretical Physics, Waterloo, Ontario N2L 2Y5, Canada
}%

\author{Sisi Zhou}%
 \email{sisi.zhou26@gmail.com}
\affiliation{%
Perimeter Institute for Theoretical Physics, Waterloo, Ontario N2L 2Y5, Canada
}%
\affiliation{%
Department of Physics and Astronomy, University of Waterloo, Waterloo, Ontario N2L 3G1, Canada
}%
\affiliation{%
Institute for Quantum Computing, University of Waterloo, Waterloo, Ontario N2L 3G1, Canada
}%
\affiliation{%
Department of Applied Mathematics, University of Waterloo, Waterloo, Ontario N2L 3G1, Canada
}

\date{\today}

\begin{abstract}
Quantum metrology aims to maximize measurement precision on quantum systems, with a wide range of applications in quantum sensing. Achieving the Heisenberg limit (HL)---the fundamental precision bound set by quantum mechanics---is often hindered by noise-induced decoherence, which typically reduces achievable precision to the standard quantum limit (SQL). While quantum error correction (QEC) can recover the HL under Markovian noise, its applicability to non-Markovian noise remains less explored.
In this work, we analyze a hidden Markov model in which a quantum probe, coupled to an inaccessible environment, undergoes joint evolution described by Lindbladian dynamics, with the inaccessible degrees of freedom serving as a memory. We derive generalized Knill-Laflamme conditions for the hidden Markov model and establish three types of sufficient conditions for achieving the HL under non-Markovian noise using QEC.
Additionally, we demonstrate the attainability of the SQL when these sufficient conditions are violated, by analytical solutions for special cases and numerical methods for general scenarios. Our results not only extend prior QEC frameworks for metrology but also provide new insights into precision limits under realistic noise conditions.
\end{abstract}

\maketitle

\tableofcontents

\section{Introduction}

Quantum metrology aims to estimate the parameters of quantum systems to the highest attainable precision. Experiments usually consist of preparing a suitable probe state, which then interacts with the system before finally measuring the probe state to obtain information on the parameter. Applications of quantum metrology include gravitational wave detection \cite{schnabel_quantum_2010}, high-precision atomic clocks \cite{katori_optical_2011}, and interferometry \cite{giovannetti_advances_2011}. 
Other promising platforms are nitrogen vacancy centers in diamond for magnetic \cite{taylor_high-sensitivity_2008,balasubramanian_nanoscale_2008} and electric \cite{dolde_electric-field_2011} field sensing, and superconducting technologies for magnetic flux sensing \cite{kleiner_superconducting_2004,danilin_quantum_2024}.

An important result in quantum metrology is the quantum Cram\'er-Rao bound \cite{helstrom_quantum_1969}, the quantum generalization of the classical bound of the same name. For a parameter $\omega$ which we wish the estimate, it bounds the variance for any unbiased estimator of $\omega$ in terms of the number of experiments $M$ and a quantity $\mathcal{F}$ called the quantum Fisher information (QFI), which is a function of the final state of the probes before measurement. 
\begin{align}\label{eq:cramer_rao}
    \lrp{\Delta \hat{\omega}}^2\geq \frac{1}{M\mathcal{F}}. 
\end{align}
This bound is saturated in the limit of many experiments ($M \gg 1$) by maximum-likelihood estimators \cite{paris_quantum_2009}. In this limit, the problem of minimizing the variance of parameter estimation is therefore equivalent to maximizing the quantum Fisher information. If we consider the task of estimating a parameter of a Hamiltonian, then the laws of quantum mechanics place a fundamental bound on the scaling of the QFI. This fundamental limit is called the Heisenberg limit (HL), and states that the QFI scaling can be no greater than a constant times $T^2$, where $T$ is the exposure time to the Hamiltonian. Equivalently, the scaling must be no greater than 
a constant times $N^2$, where $N$ is the number of probes used in a given experiment. A typical strategy for obtaining HL scaling in the number of probes is entangling $N$ probes in a GHZ state. In an experimental setting, decoherence due to noise often imposes a tighter bound on the scaling of the QFI. In such cases, the best achievable scaling is linear in $N$, and can be achieved with $N$ independent probes. This limit is called the Standard Quantum Limit (SQL).

One proposal for surpassing the SQL despite the presence of noise is to use Quantum Error Correction (QEC) \cite{arrad2014increasing,kessler_quantum_2014,dur2014improved,sekatski2017quantum,demkowicz2017adaptive,zhou_achieving_2018,zhou_asymptotic_2021,unden_quantum_2016,matsuzaki_magnetic-field_2017}. QEC is a protocol that protects quantum information from noise by encoding the information into many physical degrees of freedom, building in redundancy. This redundancy allows for the detection and correction of unwanted operations, or errors, on our state due to interaction with the environment. The principal application of QEC is quantum computation, where computations are performed on logical degrees of freedom encoded in a large number of physical degrees of freedom. It has been shown that QEC can be used to obtain HL scaling in the presence of noise so long as the noise obeys certain assumptions \cite{zhou_achieving_2018,zhou_asymptotic_2021}. One of the principal assumptions is that the noise has to be Markovian. In this work, we generalize the results of \cite{zhou_achieving_2018} and derive conditions for achieving the HL using QEC in the presence of non-Markovian noise.

There are many different ways of quantifying Markovianity in a quantum system. Examples of indicators of quantum non-Markovianity are non-CP-divisibility \cite{breuer_theory_2007,rivas_entanglement_2010,chruscinski_dynamical_2022}, information backflow \cite{breuer_measure_2009,lu_quantum_2010,chruscinski_dynamical_2022}, and tensor product structure of the process tensor \cite{milz_quantum_2021,milz_when_2020}. Non-Markovian metrology has been studied in many different contexts. It has been studied with non-Markovian bosonic environments \cite{chin_quantum_2012,macieszczak_zeno_2015}, where the non-Markovianity and the ability to perform fast measurement allow QFI scaling as $N^{3/2}$, called the Zeno limit. It has further been studied for Ramsey spectroscopy \cite{riberi_frequency_2022,xie_quantum_2014}, for sensing magnetic fields \cite{matsuzaki_magnetic_2011}, for non-equilibrium environments \cite{cheng_nonequilibrium_2023}, and for optimizing the QFI using process tensors and quantum combs \cite{altherr_quantum_2021,yang_memory_2019,kurdzialek_quantum_2024,kurdzialek_universal_2024}. Quantum error correction for non-Markovian noise, modeled by quantum combs, has previously been explored in \cite{tanggara_strategic_2024}.

Non-Markovian noise is an important source of decoherence in many experimental quantum sensing platforms. In superconducting circuits, a major source of decoherence is coupling between qubits and two-level quantum states in the environment. These two-level systems are coherent with non-negligible coherence times \cite{neeley_process_2008,shalibo_lifetime_2010,zagoskin_quantum_2006}, allowing them to serve as a memory mediating temporal correlations in the noise. Further, in many practical cases, few two-level systems are required to accurately model the noise \cite{agarwal_modelling_2024,martinis_decoherence_2005,cole_quantitative_2010}. Some theoretical models attribute the ubiquitous $1/f$ noise spectrum in solid-state systems to the presence of two-level systems. \cite{burnett_evidence_2014,mickelsen_interacting_2023}.

In nitrogen vacancy centers in diamond, for example, a major source of decoherence is coupling between the electron of the center and the surrounding carbon or nitrogen nuclear spins. The surrounding spins are coherent \cite{childress_coherent_2006}, allowing them to serve as a memory. The noise in nitrogen vacancy center experiments can also be accurately modeled using a relatively small number of these environmental spins \cite{maze_electron_2008,wang_spin_2013}.

\begin{figure}
    \centering
    \def\svgwidth{\columnwidth}
    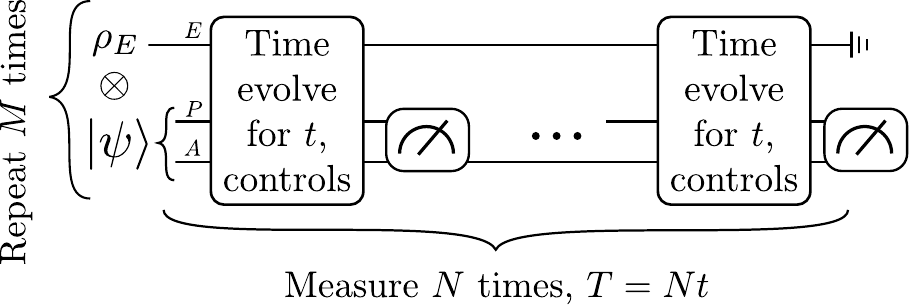
    \caption{
    Illustration of the different timescales in a sensing experiment. In a given experiment, the probe and auxiliary system are measured $N$ times, 
    each capturing the signal and the noise for a duration $t$. At this timescale, the environment is untouched. We repeat such an experiment $M$ times, where at this timescale ($T=Nt$)
    the environment is reset. 
    Note that the QEC procedure is applied fast and frequently, viewed as quantum controls during the time $t$ evolution. Although it may involve additional usage of auxiliary systems and measurements, they are not depicted explicitly here for simplicity. The auxiliary system included in the figure only represents noiseless auxiliary system that serve as part of the QECC. 
    }
    \label{fig:time_scales}
\end{figure}

\subsection{Setup}\label{subsection:def_HMM}

We consider a quantum system consisting of a probe and inaccessible degrees of freedom, which we call the environment. 
Let $\Hammy_P$ denote the $d_P$-dimensional Hilbert space of the probe, and $\Hammy_E$ represent the $d_E$-dimensional Hilbert space of the environment. Further, we suppose that we have access to noiseless auxiliary degrees of freedom. Let $\Hammy_A$ denote the Hilbert space of the auxiliary system. The total state of the quantum system is denoted by $\rho$, where $\rho \in \Hammy_E \otimes \Hammy_P\otimes \Hammy_A$. In this work, we will consider the case of finite-dimensional systems $d_E,d_P<\infty$.
The state of the probe and auxiliary system is given by the reduced density matrix $\rho_{PA} = \Tr_E(\rho)$, 
obtained by tracing out the environment. 
We suppose the state $\rho$ evolves in time according to a Lindblad master equation of the form :
    \begin{multline}\label{eq:master_eq}
        \frac{d\rho}{dt}=-\ii\lrb{H_{EP},\rho}-\ii\lrb{\omega \Tilde{G},\rho} -\ii\lrb{\omega\tilde{H}_E,\rho} \\ +\sum_{k=1}^r\lrp{L_k\rho L_k^{\dagger} -\frac{1}{2}\lrcb{L_k^{\dagger}L_k,\rho}}.
    \end{multline}
We suppose $\Tilde{G}$ is of the form $\ident_E\otimes G\otimes\ident_A$, in other words, $G$ only acts on the probe degrees of freedom. Similarly, we suppose $\tilde{H}_E$ is of the form $H_E\otimes\ident_P\otimes\ident_A$, and only acts on the environment. Further, we assume $G$ is non-trivial, i.e., $G$ is not proportional to $\ident_P$. We have that $G$, $H_E$, and $H_{EP}$ are Hermitian operators. The $r$ different $L_k$ operators are called the jump operators, also known as Lindblad operators. $H_{EP}$ and the $L_k$ act solely on the probe and environment degrees of freedom, acting trivially on the auxiliary system.

The real parameter $\omega$ is what we wish to estimate. Both $\omega\Tilde{G}$ and $\omega G$ will be referred to as the signal Hamiltonian acting on the probe, or simply the signal on the probe. The Hamiltonian $\omega H_E$ corresponds to the evolution of the environment dependent on the parameter $\omega$, or the signal on the environment. In this model, we suppose that the signal acts independently on probe and environment degrees of freedom.

The $H_{EP}$ and $L_k$ terms correspond to the undesired noise affecting the system. The effect of $H_{EP}$ is referred to as the unitary coupling. The jump operators correspond to the dissipative portion of the noise. These operators only act non-trivially on the environment and probe degrees of freedom, and act trivially on the auxiliary degrees of freedom. In this model, we suppose that the noise does not depend on the signal parameter $\omega$.

The model is Markovian when considering the system and environment together, as it evolves according to a Lindblad master equation.  
However, only the probe and auxiliary degrees of freedom are accessible experimentally. The time evolution of their reduced state exhibits non-Markovian behavior. 
For this reason, we refer to this model as a \textit{hidden Markov model} (HMM) \cite{rabiner_introduction_1986,milz_quantum_2021}. It is similar to the \textit{hidden quantum Markov model} studied in the literature \cite{monras_hidden_2012}, the principal difference here being that we consider the reduced density matrix at different times rather than measurement. When we only have access to the probe and auxiliary system's dynamics, the model displays non-Markovianity for all of the aforementioned indicators.  
In particular, the local dynamics may not be CP-divisible, and the probe may experience information backflow from the environment.

For each experiment, the environment, probe, and auxiliary system will be reset. As mentioned above, we will denote the number of experiments by $M$.
For a given experiment, we may perform arbitrary control operations on the probe and auxiliary system, while leaving the environment untouched.
The environment's degrees of freedom can therefore mediate temporal correlations in the noise perceived by the probe. We will denote by $N$ the number of times we directly measure the probe with a full-rank projective measurement. We will denote by $t$ the exposure time of each probe to the Hamiltonian between two measurements.  
During this time, we allow the application of non-destructive control operations to the probe and auxiliary system. The total experiment time is given by $MT$, where $T=Nt$. Note that we can also defer all measurements to the end of the evolution by shelving the probe state and swapping in a fresh probe from the auxiliary degrees of freedom.  
The goal is to obtain conditions for HL scaling in $T$.
\cref{fig:time_scales} illustrates these different timescales.

Finally, we assume access to an unlimited supply of noiseless auxiliary degrees of freedom, enabling the generation of perfect copies of any desired probe state. In experiment, this assumption has been well approximated by considering an electron state in a nitrogen vacancy sensor as the probe, and by using addressable $^{13}$C carbon nuclear spins, which are more robust to noise and have longer coherence times, as the auxiliary system \cite{unden_quantum_2016}.

Additionally, we assume that all control operations on both the probes and the auxiliary system can be performed with arbitrary speed and precision. For error correction to be feasible, the coherence times of our probe must greatly exceed the duration of control operations. As for the non-Markovianity, it has been shown that for certain superconducting systems interacting with environmental two-level states, the coherence times of the two-level states are of the same order as those of the qubit \cite{sun_entanglement_2012,shalibo_lifetime_2010,muller_towards_2019,neeley_process_2008,zagoskin_quantum_2006}. This would allow for non-Markovian effects to develop across many rounds of control operations.

\begin{figure}[ht]
    \centering
    \def\svgwidth{\columnwidth}
    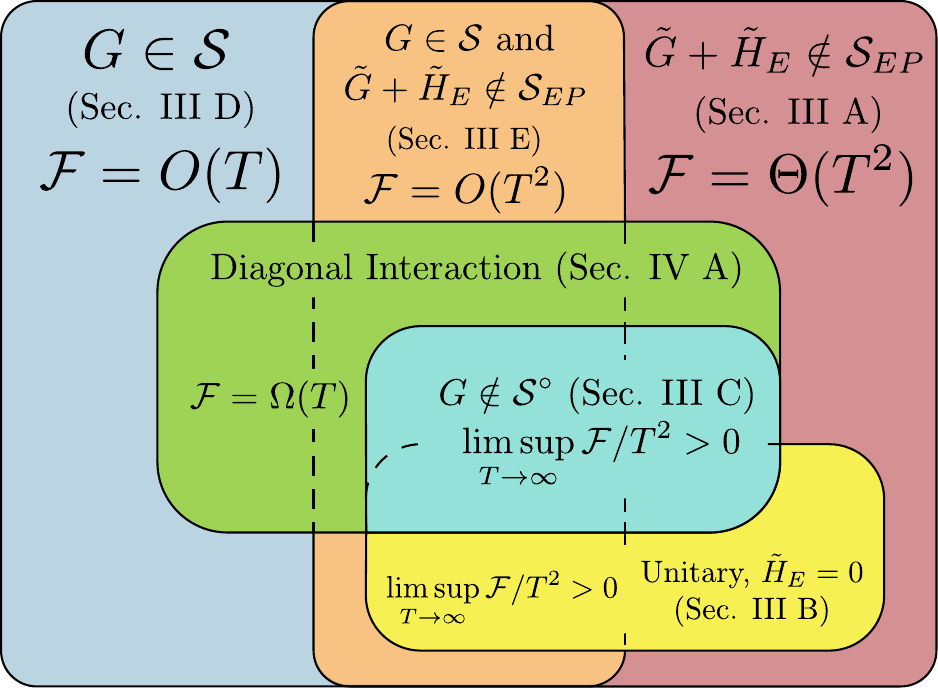
    \caption{
    Best achievable asymptotic scaling of the quantum Fisher information (QFI) for different overlaps of the signal and the noise in the hidden Markov model. The Heisenberg limit (HL) scaling in the rightmost region is guaranteed by the ``Hamiltonian not in extended span" (HNES) conditions $G\notin{\mathcal{S}}$ (See Section~\ref{subsection:HNES}). The standard quantum limit (SQL) bound in the leftmost region is guaranteed by the ``Hamiltonian not in Lindblad span" (HNLS) conditions (See Section~\ref{subsection:HNLS}). Scaling for the middle section depends on the trajectory of the evolution (See Section~\ref{subsection:e_dependence}). For the diagonal interaction, the scaling is lower bounded by the SQL (See Section~\ref{subsection:diag_sql}), or if the ``Hamiltonian not in extended Lindblad span" (HNELS) are satisfied (Section~\ref{subsection:HNELS}) then the scaling is HL up to a envelope function, as in the unitary case (See~\ref{subsection:unitary}). A list of acronyms is summarized in~\cref{app:acronyms}.
    }
    \label{fig:venn_diagram}
\end{figure}

\subsection{Summary of Results and Outline}

Consider the hidden Markov model. Let $\lrcb{\ket{\phi_i}}_{i=1}^{d_E}$ be an orthonormal basis for $\Hammy_E$. 
Let ${\mathcal{S}^{\circ}}$ be the linear span of the operators $\ident_P$, $\bra{\phi_i}L_k\ket{\phi_i},$ $\bra{\phi_i}L_k^{\dagger}\ket{\phi_i}$, and $\bra{\phi_i}L_k^{\dagger}\ket{\phi_i}\bra{\phi_m}L_j\ket{\phi_m}$  ($\forall\ k,j\in\lrcb{1,2,...,r}$, $\forall i,m\in\lrcb{1,2,...,d_E}$). Let $\mathcal{S}$ be the linear span of $\mathcal{S}^{\circ}$ where we add the operators $\bra{\phi_i}H_{EP}\ket{\phi_m}$, $\forall\ i,m\in\lrcb{1,\cdots,d_E}$. Let $\mathcal{S}_{EP}$ be the linear span of $\ident_{EP}$, $L_k$, $L_k^{\dagger}$, $L_k^{\dagger}L_j$ $\forall\ k,\ j$. Finally, we call an interaction diagonal if it is of the form presented in Definition~\ref{def:diagonal_interaction}. Using these definitions, we summarize the results of the manuscript in~\cref{tab:summary_results}.

\begin{table}[b]
  \caption{\label{tab:summary_results}
    Summary of Results.}
  \begin{ruledtabular}
    \begin{tabular}{ccc}
      Theorem & Conditions & QFI Scaling \\
      \hline
      \hline
      \ref{theorem:HNES} & $G\notin\mathcal{S}$ & $\mathcal{F}\propto T^2$ \\
      \ref{theorem:unitary_scaling} & $\tilde{H}_E=0,L_k=0\ \forall \ k$ & $\limsup_{t\rightarrow\infty}\mathcal{F}/T^2>0$ \\
      \ref{theorem:HNELS} & \specialcell{Diagonal interaction,\\$G\notin\mathcal{S}^{\circ}$} & $\limsup_{t\rightarrow\infty}\mathcal{F}/T^2>0$ \\
      \ref{theorem:HNLS_full_system} & $\tilde{H}_E+\tilde{G}\in\mathcal{S}_{EP}$ & $\mathcal{F}=O(T)$ \\
      \ref{theorem:SQL_diagonal_achievable} & Diagonal interaction & $\mathcal{F}=\Omega(T)$
    \end{tabular}
  \end{ruledtabular}
\end{table}

The manuscript is organized as follows. In Section~\ref{section:KL_HMM}, we consider quantum error correction for the hidden Markov model by adapting the Knill-Laflamme conditions.

In Section~\ref{section:achieving_HL}, the attainability of the HL using QEC is explored. The results of this section are presented in \cref{fig:venn_diagram}.
In Section~\ref{subsection:HNES}, the QEC conditions of Section~\ref{section:KL_HMM} are used to generalize the HNLS conditions of \cite{zhou_achieving_2018} to the hidden Markov model. The derived ``Hamiltonian not in extended span" (HNES) conditions (Theorem~\ref{theorem:HNES}) are sufficient conditions for achieving HL scaling in the presence of non-Markovian noise. The rightmost (red) section of \cref{fig:venn_diagram} corresponds to when these conditions are satisfied. The error correction protocol considered in this section is designed to eliminate all noise acting on the probe. As a result, the information about the parameter contained in the environment can not be exploited for potential gains in sensing.

In Section~\ref{subsection:unitary}, we consider a special case of the model for which the overall dynamics are unitary and the signal only affects the probe ($\tilde{H}_E=0$), and show that HL scaling can be achieved up to an almost periodic envelope function. This special case corresponds to the yellow bubble of \cref{fig:venn_diagram}.

In Section~\ref{subsection:HNELS}, we consider a special case of the model for which operators partially commute. In this case, a weaker set of conditions than HNES, which we call the ``Hamiltonian not in extended Lindblad span" (HNELS) conditions, guarantee that HL scaling can be achieved with QEC up to an almost periodic envelope function, similar to the unitary case. This special case corresponds to the teal bubble in the middle of \cref{fig:venn_diagram}. In both the unitary and HNELS cases, non-Markovianity plays a crucial role. A weaker QEC code is employed—one that does not correct the unitary coupling noise term, which can temporarily decohere the probe. However, since this coupling is unitary, we are guaranteed that information will flow back into the system from the environment, allowing for HL scaling.

In Section~\ref{subsection:HNLS}, the HNLS conditions of \cite{zhou_achieving_2018} are applied to rule out the HL in certain cases. This is depicted in the leftmost (blue) section of \cref{fig:venn_diagram}.

In Section~\ref{subsection:e_dependence}, an example is presented to prove that the HNES conditions are sufficient but not necessary. This example illustrates how, in certain cases, the initial environment state must be taken into account in order to derive conditions for HL versus SQL scaling. This corresponds to the middle (orange) section of \cref{fig:venn_diagram}.

The attainability of SQL scaling in $T$ is considered in Section~\ref{section:achieving_SQL}. The section is organized as follows. In Sections~\ref{subsection:diag_sql}, we consider models where the operators partially commute and show that a combination of QEC and a prepare-and-measure strategy achieves the SQL in $T$. This is depicted by the green bubble in the center of \cref{fig:venn_diagram}. Section~\ref{subsection:two_qubit_dephasing} considers a model where operators commute completely, in which case a prepare-and-measure strategy without QEC achieves SQL scaling in $T$. Once again, we ask the question of whether the SQL is achievable despite the presence of noise. We do not investigate the potential advantages of non-Markovianity for such a strategy. An example of this is the Zeno limit \cite{chin_quantum_2012,macieszczak_zeno_2015}, where the non-Markovianity allows for $N^{3/2}$ QFI scaling, between SQL and HL.

In Section~\ref{subsection:sampling_FI}, we consider a prepare-and-measure strategy without QEC for more general models, which do not necessarily have commuting operators. The Fisher information is computed numerically, providing evidence for generic SQL scaling in the case where $\tilde{H}_E=0$. In the more general case of $\tilde{H}_E\neq 0$, the numerical results are inconclusive. Numerical results are presented for a Heisenberg model-type interaction and for randomly generated master equations on two qubits.

\section{QEC for HMM}\label{section:KL_HMM}

\subsection{Non-Markovian QEC}

We begin by considering QEC for HMMs. We generalize the standard definition of a quantum error-correcting code to one that allows for errors in both the probe and the environment, while preserving the logical information in the probe through local operations only.
This can be viewed as a subsystem code where we treat the environment as the gauge degrees of freedom and the probe as the logical subsystem. The difference is that we restrict operations to act solely on the logical subsystem.
We prove the following lemma to be the generalized Knill-Laflamme (KL) conditions for the HMM.

Before deriving error correction conditions for the HMM, we must first clarify what exactly is meant by an error-correcting code. We do so in the following definition.

\begin{definition}\label{def:def_qecc}
let $U$ be an isometry from $\C^k$ to $\C^{d_{PA}}$ for some integers $k\leq d_{PA}$. 
Let $\mathcal{E}$ be a set of linear operators $\{E:\C^{d_{PA}d_E}\rightarrow \C^{d_{PA}'d_E'}\}$ where $d_E,d_E'$ are the dimensions of the environment before and after application of a linear operator.
We will say the $\lrp{U,\mathcal{E}}$ is a QECC for $\mathcal{E}$ if there exists a completely positive trace-preserving (CPTP) map $D:\mathcal{B}(\C^{d_{PA}'})\rightarrow \mathcal{B}(\C^{k})$ such that $\forall \ \ket{\psi}\in\C^k$, $\forall \ \ket{\phi}\in \C^{d_E}$ and $\forall\ E\in\mathcal{E}$,
we have that 
    \begin{align}\label{eq:def_qecc}
        (D\otimes \ident_{d_E'})\lrp{E(U\otimes \ident_{d_E})\ket{\psi}\bra{\psi}\otimes \ket{\phi}\bra{\phi}(U^{\dagger}\otimes \ident_{d_E})E^{\dagger}}\nonumber\\
       =c_{E,\psi,\phi}\ket{\psi}\bra{\psi}\otimes \alpha_{E,\psi,\phi},\quad\quad
    \end{align}
    where $\alpha_{E,\psi,\phi}$
    is a density matrix in $\C^{d'_E}\times \C^{d'_E}$.
\end{definition}

If an error $E$ satisfies \cref{eq:def_qecc}, we say the error is correctable. Here, $\C$ denotes the complex numbers and $\mathcal{B}(\C^{d})$ denotes the set of linear operators on the vector space $\C^d$. 

QECC are usually characterized using the KL conditions \cite{knill_theory_1996}, which are both necessary and sufficient. Lemma~\ref{lemma:KL} presents similar conditions for QECC for HMMs. 

\begin{lemma}[KL Conditions for the HMM]\label{lemma:KL}
Let $U$ be an isometry from $\C^k$ to $\Hammy_{PA}=\Hammy_P\otimes\Hammy_A$, of dimension $d_{PA}$, and let $\mathcal{E}$ be a set of linear operators mapping $\Hammy_E\otimes\Hammy_{PA}$ to $\Hammy_{E'}\otimes\Hammy_{PA'}$. Let $P=UU^{\dagger}$ be the code space projector and let $\lrcb{\ket{i}_E}_{i=1}^{d_E}$ $(\tx{resp. }\lrcb{\ket{i}_{E'}}_{i=1}^{d_E'})$ be an orthonormal basis for $\Hammy_E$ (resp. $\Hammy_{E'}$). 
We have that $\lrp{U,\mathcal{E}}$ is a quantum error correcting code (QECC) if and only if we have
    \begin{align}\label{eq:KLM_HMM}
P\bra{l}_{E}E_i^{\dagger}\ket{m}_{E'}\bra{n}_{E'}E_j\ket{k}_{E} P=c_{ij}^{lmnk}P,
    \end{align}
    for all $E_i,E_j\in\mathcal{E}$, where $c_{ij}^{lmnk}$ are scalar constants.
\end{lemma}

The proof of Lemma~\ref{lemma:KL} can be found in Appendix~\ref{section:KL_conditions_proof}. The key insight of the proof is that by defining the set of extended errors,
\begin{equation}\label{eq:extended_errors_main}
\Tilde{E}_{i,n,k}=\bra{n}_{E'}E_i\ket{k}_{E},
\end{equation}
the non-Markovian KL conditions reduce to the standard KL conditions on this extended error set. This allows us to reuse all of the error correcting machinery, including encoding maps and recovery operations, from the standard theory.



    Consider the case where $d_{PA}=d_{PA}'$ and $d_E=d'_E$, which captures our model. The KL conditions and Definition~\ref{def:def_qecc} ensure that the state is always a product state on $\Hammy_E\otimes\Hammy_{PA}$. Therefore, the time-evolution of the reduced state on $\Hammy_{PA}$ between two rounds of error correction can always be represented by a CPTP map which depends on the residual states on $\Hammy_E$. The KL conditions ensure that the Kraus operators of this map will be correctable for any choice of the state on $\Hammy_E$. Further, implementing the error correction protocol requires no knowledge of the environment state.
    
    Note that this error-correcting condition for a QECC is perhaps too strong. 
    In principle, the capabilities to correct errors induced by any state on $\Hammy_E$ are not necessary; one needs only to correct the errors induced by the states on $\Hammy_E$ occurring throughout the evolution of the system. This distinction will be important when we consider quantum metrology. We obtain sufficient but not necessary conditions for the HL for this reason.

    Finally, an important remark is that even when perfect QEC fails, the recovery operation still brings the state back to the code space. This will be important for the derivations of Sections~\ref{subsection:HNELS} and~\ref{subsection:diag_sql}, where even if QEC doesn't completely correct the noise, it tailors it to a specific form and restricts it to the code space.

    Other non-Markovian quantum error correction conditions exist in the literature. In \cite{tanggara_strategic_2024}, the authors derive necessary and sufficient conditions for quantum error correction for quantum combs, which they call the \textit{strategic code formalism}. The model they consider is very general, and encapsulates the HMM defined in Section~\ref{subsection:def_HMM}. However, since we are considering a simpler model, we can derive conditions using the operators of the master equation directly, rather than considering different contractions of the quantum comb with the QEC control operations.

\subsection{A Small Example}

We will now consider an example of quantum error correction in a spin system to highlight the differences between the Markovian and non-Markovian cases.
We will consider a system where $\Hammy_P$ consists of an ensemble of individually addressable spin-half degrees of freedom (or qubits), and the environment $\Hammy_E$ consists of inaccessible spin-1 degrees of freedom. For simplicity, we will assume each spin-half particle interacts with its own bath of spin-1 particles. In the non-Markovian case, we will assume this bath consists of a single coherent spin. We will suppose the coupling Hamiltonian between a spin-half particle and the $i$-th spin-1 particle in the system is given by
\begin{equation*}
    H_{EP,i}=\tau_i\lrp{S_z^{(i)}\otimes Z + S_x^{(i)}\otimes X},
\end{equation*}
where $Z,X$ are the Pauli matrices acting on the spin-half particle and $S_x^{(i)},S_z^{(i)}$ are the spin-1 operators acting on the $i$-th spin-1. $\tau_i$ is the coupling strength.
\begin{equation*}
    S_z^{(i)}=\begin{pmatrix}
        1 & 0 & 0 \\
        0 & 0 & 0 \\
        0 & 0 & -1
    \end{pmatrix},\quad S_x^{(i)}=\frac{1}{\sqrt{2}}\begin{pmatrix}
        0 & 1 & 0 \\
        1 & 0 & 1 \\
        0 & 1 & 0
    \end{pmatrix}.
\end{equation*}
Let $\ket{-1}_E,\ket{0}_E$, and $\ket{1}_E$ be the $-1,0,$ and 1 eigenstates of $S_z$ respectively. We will suppose that each spin-1 has a Hamiltonian term of the form 
\begin{equation*}
    H_{E,i}=\Delta\lrp{S_z^{(i)}}^2,
\end{equation*}
with $\Delta$ a positive strength constant. The total Hamiltonian for a spin-half and its bath is given by 
\begin{equation*}
    H_{tot.}=\sum_{i}\lrp{H_{E,i}+H_{EP,i}}.
\end{equation*}

We first consider the Markovian regime. In this regime, we assume each spin-half interacts with many spin-1 degrees of freedom. Following the prescription of \cite{breuer_theory_2007}, we make two approximations. We first suppose that the spin-half degrees of freedom couple weakly to the bath of environmental spins, such that the bath is weakly affected by the interaction ($\Delta\gg \tau_i$). Further, we suppose there is a dissipative master equation acting on the environment such that the spin-1 degrees of freedom are driven to the $\ket{0}_E$ steady state. A Markovian master equation can then be derived for the spin-half degrees of freedom. This is done in Appendix~\ref{section:example_markov_approx}. This master equation generates a noise channel with Kraus operators proportional to $\ident$ and $X$.
Therefore, a repetition code \cite{peres_reversible_1985,albert_quantum_2022} that only corrects $X$ errors is sufficient to protect the logical information.

We next consider the non-Markovian regime. In this regime, we suppose the spin-half degrees of freedom each interact with a single coherent spin-1 degree of freedom. For simplicity, we suppose that there is no dissipation. In the non-Markovian regime, since the environment state evolves in time, we must consider the set of extended errors
\begin{equation*}
    K_{i,j}=\bra{i}_Ee^{-\ii t H_{tot.}}\ket{j}_E,
\end{equation*}
where $\ket{i}_E, \ket{j}_E\in \{\ket{-1}_E,\ket{0}_E,\ket{1}_E\}$.
Since we consider arbitrarily fast control, we can consider the errors to first order in $t$. We obtain the operators
\begin{equation*}
    K_{0,0}=\ident, \quad K_{j,j}=\ident -\ii t \Delta\ident  - \ii j t \tau Z,
\end{equation*}
\begin{equation*}
    K_{-1,0}=K_{0,-1}=K_{1,0}=K_{0,1}=\frac{-\ii \tau  t}{\sqrt{2}}X.
\end{equation*}
We see that to protect the logical information, a repetition code no longer suffices due to the presence of $Z$ rotation errors. Because of the additional error operators introduced by the non-Markovianity, a more capable code is required.
    
\subsection{Computation versus Metrology}

There are fundamental differences between the tasks of quantum sensing and quantum computation that must be taken into account when designing a QECC. The principal difference relates to the logical operators, which are the operators that act non-trivially on the code words without being detected. In quantum computation, the goal is to design codes with logical operators that do not resemble the noise. Logical operators are engineered so that the errors induced by the environment have a very low probability of corrupting the logical information. There is much freedom in engineering these logical operators, allowing for the design of complex codes with very high code distance.

In quantum sensing, we wish to do the opposite. We wish to design codes for which the interaction with the environment due to the signal is projected to a non-trivial logical operator. Contrary to the case of computation, these interactions are not freely chosen. There is a lot less freedom when engineering a QEC protocol for quantum sensing. A trade-off between preserving the signal and eliminating noise must be achieved.

As previously shown, the non-Markovian quantum error correction conditions for the hidden Markov model can be rephrased as the standard Knill-Laflamme conditions acting on an extended error set defined by \cref{eq:extended_errors}. In other words, the non-Markovianity of the system introduces a larger variety of errors. Consider the previous example of each qubit interacting with a spin-1 particle. The non-Markovianity changes the form of the errors to either $X$ or $Z$ operators rather than just $X$ operators.

Due to the restricted design freedom, this problem is much more harmful in quantum sensing. Consider the same example, where we wish to use the spin-1/2 particles to detect a magnetic field. Suppose the magnetic field acts independently on each spin as a $Z$ operator. We therefore design a code for which $Z$ is a logical operator. In the Markovian case, the errors are of the form $X$ and are therefore distinguishable from $Z$, and so a repetition code does the job. In the non-Markovian case, we also have $Z$ errors, which have the same form as the magnetic field, and so the logical information will always be corrupted, no matter which types of QECC are used.

To show why this may not be an issue for computation, we consider the example of a stabilizer code. Stabilizer codes have three important parameters : the number of physical qubits $n$, the number of logical qubits $k$ encoded in the physical qubits, and finally the distance $d$. A distance d stabilizer code can correct error operators which can be written as products of $\floor{\frac{d-1}{2}}$ arbitrary single qubit errors~\cite{gottesman_surviving_2024}. For unbiased noise, the non-Markovianity will therefore be problematic for quantum error correction if the memory effects in the system cause the errors to be high-weight. When this is not the case and when the non-Markovian noise is also local, it has been shown that a fault-tolerance threshold is still possible~\cite{terhal_fault-tolerant_2005,aliferis_quantum_2005}. In the previous example, the non-Markovianity does not affect the weight of the errors.

For the experimental platforms and corresponding noise sources previously mentioned, this is unlikely to be a problem. For example, in superconducting qubits, one expects the interaction with the two-level chip defects to be relatively local. Even though the memory and coherence of the defect may change the form of the error, it will not affect the weight of the error. Major sources of error in superconducting experiments have instead been crosstalk or high-weight errors due to cosmic rays \cite{acharya_suppressing_2023,mcewen_resolving_2022}.

For quantum sensing and metrology, however, non-Markovianity has the potential to be much more destructive, as we will see below.

\section{Achieving the HL}\label{section:achieving_HL}

\subsection{The HNES Conditions for Achievability of the HL}
\label{subsection:HNES}

In \cite{zhou_achieving_2018}, the authors present sufficient and necessary conditions, called ``Hamiltonian not in Lindblad span" (HNLS), for achieving the HL asymptotically. 
When these conditions are met, the probe and a noiseless auxiliary system can be encoded into a specialized error-correcting code. For this code, all of the noise becomes correctable errors, and a portion of the signal remains as a logical operation.
This logical component of the signal leads to HL scaling within the codespace. 
However, the HNLS conditions were originally derived for Markovian noise, where no environmental degrees of freedom are present, and the master equation applies solely to the probe. Further, these conditions are necessary to achieve the HL, as non-correctable errors inevitably lead to complete decoherence of the system and the signal over long timescales.

We generalize these conditions to the HMM. The HNLS conditions were derived under the assumption that error correction can be performed perfectly and arbitrarily fast. We will keep the same assumptions. 

\begin{theorem}[Hamiltonian not in extended span (HNES) Conditions]\label{theorem:HNES}
Let $\lrcb{\ket{\phi_i}}_{i=1}^{d_E}$ be an orthonormal basis for $\Hammy_E$. Let $\mathcal{S}$ be the linear span of the operators $\ident_P$, $\bra{\phi_i}H_{EP}\ket{\phi_m}$, $\bra{\phi_i}L_k\ket{\phi_m}$, $\bra{\phi_i}L_k^{\dagger}\ket{\phi_m}$ and $\bra{\phi_i}L_k^{\dagger}\ket{\phi_m}\bra{\phi_n}L_{j}\ket{\phi_l}$ ($\forall\ k,j\in\lrcb{1,\cdots,r}$, where $r$ is the number of jump operators; $i,m,n,l\in\lrcb{1,\cdots,d_E}$).

If $G\notin\mathcal{S}$, then HL scaling can be achieved for the parameter $\omega$. The condition $G\notin\mathcal{S}$ will be referred to as the ``Hamiltonian not in extended span'' (HNES) conditions. The extended span refers to the possible errors that can affect the reduced state of the probe for arbitrary choices of the environment state.
\end{theorem}

\begin{proof}
The proof generalizes the HNLS conditions for Heisenberg-limit scaling from \cite{zhou_achieving_2018}, with modifications to accommodate the HMM. It will therefore be very similar to the original proof, with some slight changes.
The evolution of the state $\rho$ for a duration $\Delta t$
is given by the first-order expansion of the master equation
    \begin{align}
        &\Phi_{\Delta t}(\rho)=\rho+\lrp{-\ii\lrb{H_{EP},\rho}-\ii\lrb{\omega \Tilde{G},\rho} -\ii\lrb{\omega\tilde{H}_E,\rho}}\Delta t \nonumber\\
        &\qquad +\Delta t\sum_{k=1}^r\lrp{L_k\rho L_k^{\dagger}-\frac{1}{2}\lrcb{L_k^{\dagger}L_k,\rho}}+O(\Delta t^2).
    \end{align}
To reduce clutter, we use the following notation
    \begin{equation}\label{eq:opertor_ij}
        A^{(i,j)}:=\bra{\phi_i}A\ket{\phi_j},
    \end{equation}
    where $A$ is a linear operator on $\Hammy_E\otimes \Hammy_P$. This leaves us with an operator acting on $P$. 
    Let $\mathcal{S}_H$ denote the subspace of Hermitian matrices in $\mathcal{S}$, which is spanned by $L_k^{(i,j)}+{L_{k}^{(i,j)}}^{\dagger}$, $\ii\lrp{L_k^{(i,j)}-{L_k^{(i,j)}}^{\dagger}}$, ${\lrp{L_k^{\dagger}L_m}}^{(i,j)}+h.c.$, $\ii\lrp{\lrp{{L_k}^{\dagger}L_m}^{(i,j)}-h.c.}$, $\ii\lrp{H_{EP}^{(i,j)}-{H_{EP}^{(i,j)}}^{\dagger}}$, and $H_{EP}^{(i,j)}+h.c.$. 
    Remark that $\mathcal{S}_H\subseteq \mathcal{S}$. $G$ is Hermitian and can therefore be decomposed into $G=G_{\parallel}+G_{\perp}$ where $G_{\parallel}\in \mathcal{S}_{H}$ and $G_{\perp}\perp \mathcal{S}_{H}$. 
    The assumption that $G\notin \mathcal{S}$ implies that $G_{\perp}$ is nonzero. $G_{\perp}$ must have trace zero in order to be orthogonal to $\ident_P$ under the trace inner product. The spectral decomposition for $G_{\perp}$ is therefore 
    \begin{equation*}
        G_{\perp}=\frac{1}{2}\lrp{\tr\abs{G_{\perp}}}\lrp{\rho_0-\rho_1}
    \end{equation*}
    where $\abs{G_{\perp}}:=\sqrt{G_{\perp}^2}$, and $\rho_0$ and $\rho_1$ are density matrices. We consider $\ket{C_0}$ and $\ket{C_1}$ to be normalized purifications of $\rho_0$ and $\rho_1$ on $\Hammy_P\otimes \Hammy_A$ where $\Hammy_A$ consists of auxiliary degrees of freedom that we append to our system. We further construct $\ket{C_0}$ and $\ket{C_1}$ to have orthonormal supports on $\Hammy_A$. To be explicit, assuming $\rho_{0} = \sum_{i=1}^{d_1} (\rho_0)_{ii} \ket{i}_P\bra{i}_P$ and $\rho_{1} = \sum_{i=d_1+1}^{d_P} (\rho_1)_{ii} \ket{i}_P\bra{i}_P$ are in the diagonalized forms, 
    \begin{gather*}
        \ket{C_0} = \sum_{i=1}^{d_1} \sqrt{(\rho_0)_{ii}} \ket{i}_P\ket{i}_A,\\ 
        \ket{C_1} = \sum_{i=d_1+1}^{d_P} \sqrt{(\rho_1)_{ii}} \ket{i}_P\ket{i}_A,
    \end{gather*}
    where $\{\ket{i}_P\}$ and $\{\ket{i}_A\}$ are orthonormal bases of the probe and the auxiliary system, respectively, where the auxiliary system is assumed to be as large as the probe. The support of $\rho_0$ and $\rho_1$ are ${\rm span}\{\ket{i}_P\}_{i=1}^{d_1}$ and ${\rm span}\{\ket{i}_P\}_{i=d_1+1}^{d_P}$ respectively. 
    This will be the code space, with the code space projector given by 
    \begin{equation}
        P=\ket{C_0}\bra{C_0}+\ket{C_1}\bra{C_1}.
    \end{equation}
    The orthogonality of the states on the auxiliary system space implies that 
    \begin{equation}\label{eq:Lindblad_proof_cross}
        \bra{C_0}B\otimes\ident_A\ket{C_1}=\bra{C_1}B\otimes\ident_A\ket{C_0} = 0
    \end{equation}
    where $B$ is any linear operator on $\Hammy_P$. 
    We also use the relation :
    \begin{equation}\label{eq:Lindblad_proof_1}
        \tr\lrp{\lrp{\ket{C_0}\bra{C_0}-\ket{C_1}\bra{C_1}}\lrp{B\otimes\ident_A}}=\frac{2\tr\lrp{G_{\perp}B}}{\tr\abs{G_{\perp}}}.
    \end{equation}
    Although a smaller auxiliary system, such as a qubit, is enough to guarantee~\cref{eq:Lindblad_proof_cross}, it would not be enough to guarantee~\cref{eq:Lindblad_proof_1}, which usually requires $d_A$ dimensions.

    \cref{eq:Lindblad_proof_1} is zero for any $B\in \mathcal{S}$, allowing us to write 
    \begin{equation*}
        \bra{C_0}B\otimes\ident_A\ket{C_0}=\bra{C_1}B\otimes\ident_A\ket{C_1}.
    \end{equation*}
    These equations imply that 
    \begin{equation*}
        PL_k^{(i,j)}P=\lambda_k^{ij}P,
    \end{equation*}
    \begin{equation*}
        P\lrp{L_k^{\dagger}L_m}^{(i,j)}P=\mu_{km}^{ij}P,
    \end{equation*}
    where this equation is due to the fact that $\lrp{L_k^{\dagger}L_m}^{(i,j)}$ is in the linear span of ${L_k^{(i,j)}}^{\dagger}{L_{k'}^{(i',j')}}$.
    By assumption on the span $\mathcal{S}$, we also have that 
    \begin{equation}\label{eq:Lindblad_kl}
        P{L_k^{\dagger}}^{(i,j)}L_m^{(i',j')}P=\zeta_{km}^{iji'j'} P.
    \end{equation}
    Further,
    \begin{equation}
        PH_{EP}^{(i,j)}P=\chi^{ij}P.
    \end{equation}
    We see from \cref{eq:Lindblad_proof_1} that 
    \begin{equation}\label{eq:PGPconst}
        PGP\neq \tx{const.}\times P.
    \end{equation}
    Let $\rho=\rho_E\otimes\rho_{PA}$ be a state in the code space,
    i.e. $P\rho_{PA}P=\rho_{PA}$. Let $\rho_E=\sum_{i=0}^{d_E}\alpha_i\ket{\phi_i}\bra{\phi_i}$. We can take $\rho_E$ to be diagonal without loss of generality since the choice of basis $\lrcb{\ket{\phi_i}}_{i=1}^{d_E}$ is arbitrary. We will use $P$ and $\ident_E\otimes P$ interchangeably; it will be clear from the context which operator we are referring to.
    \begin{multline*}
            P\Phi_{\Delta t}(\rho)P=\rho-\ii\omega\lrb{P\Tilde{G}P,\rho}\Delta t -\ii\omega\lrb{H_E,\rho_E}\otimes\rho_{PA}\, \Delta t \\ + \lrp{\sum_{i}\alpha_i \sum_{k=1}^r A_{ik} \otimes \rho_{PA}} \Delta t + O(\Delta t^2), 
       \end{multline*}
    where 
    \begin{multline*}
        A_{ik} := \sum_{i',j'}\lambda_k^{i'i}\lrp{\lambda_k^{j'i}}^*\ket{\phi_{i'}}\bra{\phi_{j'}} \\+ 
        \sum_{i'}\Big((\mu_{kk}^{i'i}-\ii\chi^{i'i})\ket{\phi_{i'}}\bra{\phi_{i}}+h.c.\Big). 
    \end{multline*}
    We consider as well the projector $P_{\perp}=\ident-P$. We have that 
    \begin{multline}\label{eq:corrected_part}
        P_{\perp}\Phi_{\Delta t}(\rho)P_{\perp}= \\ \Delta t \sum_{i}\alpha_i \sum_{k=1}^r  B_{ik} \lrp{\ket{\phi_i}\bra{\phi_i}\otimes \rho_{PA}} B_{ik}^\dagger + O(\Delta t^2),
    \end{multline}
    where $B_{ik} = \sum_j\ket{\phi_j}\bra{\phi_i}\otimes \lrp{L_k^{(j,i)}-\lambda_k^{ji}\ident_{PA}}$.
    After measurement and neglecting $\Delta t^2$ terms, the state is given by
    \begin{multline*}
        P\Phi_{\Delta t}(\rho)P+P_{\perp}\Phi_{\Delta t}(\rho)P_{\perp}= 
        \rho-\ii\omega\lrb{P\Tilde{G}P,\rho}\Delta t\\ -\ii\omega\lrb{H_E,\rho_E}\otimes\rho_{PA}\, \Delta t  + \lrp{\sum_{i}\alpha_i \sum_{k=1}^r A_{ik} \otimes \rho_{PA}} \Delta t \\
        +\Delta t \sum_{i}\alpha_i \sum_{k=1}^r  B_{ik} \lrp{\ket{\phi_i}\bra{\phi_i}\otimes \rho_{PA}} B_{ik}^\dagger.
    \end{multline*}
    Except for the final term, all noise contributions can be effectively suppressed by fast, repeated measurements, which project the system onto the code space—an instance of the quantum Zeno effect. However, the operators $B_{ik}$ in the final term act on the probe, introducing dephasing that cannot be mitigated by measurement alone. Hence, measurement is not enough, and the recovery operation is required here as well.
    
    
    \cref{eq:Lindblad_kl} is exactly that of the generalized KL conditions for the HMM in Lemma~\ref{lemma:KL}, where we take the errors to be the linear operators acting on the state in \cref{eq:corrected_part}. Therefore, there exists a recovery operation $\mathcal{R}$ such that
    \begin{equation}\label{eq:recovery_perp}
        \mathcal{R}\lrp{P_{\perp}\Phi_{\Delta t}(\rho)P_{\perp}}= \Delta t\,  \sigma_E\otimes \rho_{PA} + O(\Delta t^2),
    \end{equation}
    where $\Delta t\sigma_E$ is a residual operator supported on the environment degrees of freedom after applying the recovery operation to \cref{eq:recovery_perp}.

    Combining \cref{eq:recovery_codespace} and \cref{eq:recovery_perp}, we obtain the final state after QEC :
    \begin{multline}\label{eq:recovery_codespace}
        P \Phi_{\Delta t}(\rho) P+ \mathcal{R}\lrp{P_\perp\Phi_{\Delta t}(\rho)P_\perp}= \rho -\ii\omega\lrb{P\tilde{G}P,\rho}\Delta t\\  
        -\ii\omega \lrb{H_E,\rho_E}\otimes\rho_{PA}\Delta t + \Delta t\sigma_E'\otimes\rho_{PA} + O(\Delta t^2),
    \end{multline}
    where $\sigma_E'=\sigma_E+\sum_i\alpha_i\sum_{k=1}^r A_{ik}$.  $\sigma_E'\otimes\rho_{PA}$ can also be written as
    \begin{align}\label{eq:residual_state}
        &\sigma'_E\otimes\rho_{PA}=\nonumber \\
        &\mathcal{D}\lrp{-i\lrb{H_{EP},\rho}+\sum_{k=1}^r\lrp{L_k\rho L_k^{\dagger}-\frac{1}{2}\lrcb{L_k^{\dagger}L_k,\rho}}}
    \end{align}
    where $\mathcal{D}(\cdot)=P\cdot P+\mathcal{R}(P_{\perp}\cdot P_{\perp})$. Since the argument of the RHS of \cref{eq:residual_state} has trace zero, then $\sigma_E'$ must have trace zero due to the fact that $\mathcal{D}$ is a CPTP map and that $\Tr\lrp{\rho_{PA}}=1$.
    We trace out the environment to obtain the reduced dynamics on the probe, which are unitary,
    \begin{multline}
        \tr_E\!\lrb{\mathcal{R}\lrp{P_{\perp}\Phi_{\Delta t}(\rho)P_{\perp}}+P\Phi_{\Delta t}(\rho)P}
        =\\\rho_{PA}-\ii\omega\lrb{P \lrp{G\otimes\ident_A} P, \rho_{PA}}\Delta t.
    \end{multline}
    We also neglect the second-order terms in $\Delta t$. From here on out, the proof is equivalent to the one presented in \cite{zhou_achieving_2018}. We summarize the final steps. Let $\lambda_0,\lambda_1$ be the eigenvalues of $P\lrp{G\otimes \ident_A}P$ associated with the eigenstates $\ket{C_0}$ and $\ket{C_1}$ respectively. We can see that the codewords are eigenstates of the effective Hamiltonian from \cref{eq:Lindblad_proof_cross}. By choosing an initial state of 
    \begin{equation*}
        \ket{\psi}=\frac{1}{\sqrt{2}}\lrp{\ket{C_0}+\ket{C_1}}
    \end{equation*}
    on the probe and auxiliary system, one can show that the QFI scales as
    \begin{equation*}
        \mathcal{F}(\rho_{PA}(t))=t^2\lrp{\lambda_{0}-\lambda_{1}}^2.
    \end{equation*}
    \cref{eq:PGPconst} guarantees that this QFI is nonzero for $t>0$. We therefore have that if the HNES conditions are satisfied, the QFI scales as a constant multiplied by $t^2$, and so we obtain the HL in $T$.\end{proof}

Note that the above QEC protocol to achieve the HL assumes QEC, as quantum controls, is performed fast and frequently throughout the state evolution. It corresponds to the situation in \cref{fig:time_scales} where the probe is only measured at the end of the evolution, where $N = 1$ and $T = t$. 


Further, although the above protocol yields HL scaling, it is not optimal in general. In particular, the constant factors in the expression of the QFI may be suboptimal. To achieve better scaling, the code may be optimized in the same manner as \cite{zhou_asymptotic_2021}.

Finally, from the point of view of information backflow, this error correction protocol eliminates all non-Markovianity in the system. Since the QEC protocol ensures that the system and environment are always in a product state, information backflow from the environment is not possible. Although this prevents potential correlations in the system from introducing harmful effects, it also prevents any potential sensing advantages due to the non-Markovianity, such as, for example, using the environment as a memory protected from decoherence \cite{zagoskin_quantum_2006}.

\subsection{Unitary Evolution}\label{subsection:unitary}

We now consider the special case of when the evolution of the system is unitary. We will also consider the additional assumption of the signal acting only on the probe, $\tilde{H}_E=0$. In this scenario, the master equation governing the evolution of the HMM (\cref{eq:master_eq}) reduces to 
\begin{equation}\label{eq:master_eq_unitary}
    \frac{d\rho}{dt}=-\ii\lrb{H_{EP},\rho}-\ii\lrb{\omega \Tilde{G},\rho}.
\end{equation}
When the HNES conditions are satisfied, the QFI scales as $t^2$. For unitary evolution with the signal acting only on the probe, we can guarantee that the QFI scales as a constant times $t^2$ multiplied by some periodic envelope function, even when the HNES conditions are not satisfied. Once again, let $N=1$. In particular, we can show that $\limsup_{t\rightarrow\infty}\mathcal{F}/t^2>0$ when $d_E<\infty$. We can therefore estimate the parameter with $t^{-2}$ precision as long as we make sure to measure the system at the times when the value of the envelope is above some nonzero constant. 

\begin{theorem}\label{theorem:unitary_scaling}
    Suppose the evolution of the system is unitary and that the signal doesn't act on the environment, in other words, the evolution of the HMM is generated by \cref{eq:master_eq_unitary}. Suppose $\Hammy_E$ has finite dimension, i.e., $d_E<\infty$. Under these assumptions, there is a QEC protocol that achieves $\limsup_{t\rightarrow\infty}\mathcal{F}/t^2>0$. More precisely, it is possible to achieve a QFI that scales as $t^2$ multiplied by an almost periodic envelope function. By measuring at correct times, we recover the HL. 
\end{theorem}

\begin{proof}

We begin by defining a QEC code in the same manner as in the proof of Theorem~\ref{theorem:HNES} with the extended span $\mS$ replaced by ${\rm span}\{\ident\}$. $\ket{C_0}$ and $\ket{C_1}$ in this case correspond to the positive and negative parts of $G - \frac{\trace(G)}{d_P}\ident_P$ (the eigenstates with eigenvalues zero can be assigned in either the positive or negative parts without affecting the calculations), whose supports are denoted by projectors $P_0$ and $P_1$. $P_0 + P_1 = \ident_{P} = \ident_{A}$. To guarantee HL scaling with this QEC code, it is not necessary to perform the recovery operation. Fast and accurate syndrome measurements will be sufficient. The QEC serves to transform the noise and enforce commutation with the signal, facilitating the QFI computation.

Note that since our only controls consist of rapid measurements, the protocol can be seen as an application of the quantum Zeno effect. By measuring fast and frequently, we freeze the dynamics relating to any noise that is orthogonal to the non-trivial component of the signal.

For an initial state in the code space, the resulting dynamics are governed by the master equation
\begin{align}\label{eq:P_Lindblad}
    \frac{d\rho}{dt} &= \lim_{\Delta t\rightarrow 0}\frac{P \Phi_{\Delta t}(\rho) P + P_\perp \Phi_{\Delta t}(\rho) P_\perp}{\Delta t} \\ \nonumber
    &= -\ii\lrb{PH_{EP}P,\rho}-\ii\lrb{\omega P\tilde{G}P,\rho},
\end{align}
where $P$ is the projection operator onto the code space. Expanding the evolution to first order, we obtain
\begin{align*}
        &P \Phi_{\Delta t}(\rho) P + P_\perp \Phi_{\Delta t}(\rho) P_\perp\\
        =&\rho-\ii\lrb{P(\omega\Tilde{G} + H_{EP}) P,\rho}\Delta t  + O(\Delta t^2).
\end{align*}
We neglect the second-order terms in $\Delta t$. This master equation corresponds to unitary evolution generated by the following effective Hamiltonian 
\begin{align*}
    H_\text{eff}=\omega P\Tilde{G}P+PH_{EP}P.
\end{align*}
First, we show that the two terms of the Hamiltonian commute. Since, from \cref{eq:Lindblad_proof_cross}, $PB_{EP}P = P_0B_{EP}P_0+P_1B_{EP}P_1$ for any operator $B_{EP}$ acting on the probe and the environment, we have 
\begin{align}
    &\lrb{PH_{EP}P,P\tilde{G}P}\\\nonumber
    =&\lrb{P_0 H_{EP}P_0,P_0\tilde{G}P_0} + \lrb{P_1 H_{EP}P_1,P_1\tilde{G}P_1}
\end{align}
where $P_0 = \ket{C_0}\bra{C_0}$ and $P_1 = \ket{C_1}\bra{C_1}$. 
Since $\tilde G$ only has support on the probe degrees of freedom, we have that $\lrb{\bra{C_i}\Tilde{G}\ket{C_i},\bra{C_j}H_{EP}\ket{C_j}}=0$ $\forall\ i,j$. Therefore, the terms cancel and we obtain 
\begin{equation*}
    \lrb{PH_{EP}P,P\Tilde{G}P}=0.
\end{equation*}
We will now show that the HL is achieved. We prepare the following state on the probe and auxiliary system
\begin{equation*}
    \ket{\psi_{PA}}=\frac{1}{\sqrt{2}}\lrp{\ket{C_0}+\ket{C_1}},
\end{equation*}
where $\ket{C_0},\ket{C_1}$ are the code words of the error correcting code.
Let $\ket{\psi_E}$ be an arbitrary state on the environment degrees of freedom. We can take $\ket{\psi_E}$ to be pure without loss of generality by considering a purification of the environment. The total state of the system is
\begin{equation}
    \ket{\Psi_{PAE}}=\ket{\psi_E}\frac{1}{\sqrt{2}}\lrp{\ket{C_0}+\ket{C_1}}.
\end{equation}
We can rewrite this state as a sum over the joint eigenstates of $PH_{EP}P\otimes\ident_A$ and $\ident_E\otimes P\Tilde{G}P\otimes\ident_A$
\begin{equation}
\label{eq:product}
    \ket{\Psi_{EPA}}=\sum_{k=1}^{d_E} c_k^0 \ket{\psi_k^{0}}\ket{C_0}+\sum_{l=1}^{d_E} c_l^1 \ket{\psi_l^{1}}\ket{C_1},
\end{equation}
where the joint eigenstates of these operators are necessarily product states between the codewords and states $\ket{\psi_i^j}$ (defined below) on the environment. Below, we will drop the tensor products with identity on the environment and auxiliary system from the notation to reduce clutter; it will be clear from the context which operators act on which subsystems. To prove \cref{eq:product} follows, we know from \cref{eq:Lindblad_proof_cross} that 
\begin{equation}
    \bra{C_0}PH_{EP}P\ket{C_1}=\bra{C_0}P\Tilde{G}P\ket{C_1}=0.
\end{equation}
This allows us to write the terms of the effective Hamiltonian in the following form 
\begin{equation*}
    PH_{EP}P=H_{EP}^{0}\otimes\ket{C_0}\bra{C_0}+H_{EP}^{1}\otimes\ket{C_1}\bra{C_1},
\end{equation*}
\begin{equation*}
    P\Tilde{G}P=\lambda_0\ident\otimes\ket{C_0}\bra{C_0}+\lambda_1\ident\otimes\ket{C_1}\bra{C_1},
\end{equation*}
where $\lambda_0,\lambda_1$ are the (non-trivial) eigenvalues of $P\tilde{G}P$ as defined in the previous section and $H_{EP}^0,H_{EP}^1$ are operators acting on the environment. From this form, it is clear that the joint eigenvectors of $PH_{EP}P$ and $P\tilde{G}P$ are the tensor products between the eigenvectors of the blocks $H_{EP}^i$ and the projectors $\ket{C_i}\bra{C_i}$.

We now consider the time evolution for our chosen probe state. Since the terms of the Hamiltonian commute, we can independently consider their generated time evolution. We begin with the signal 
\begin{align*}
    e^{-\ii t\omega P\Tilde{G}Pt}\ket{\Psi_{EPA}}
    =\sum_{k=1}^{d_E} c_k^0 e^{-\ii t\omega \lambda_0} \ket{\psi_k^0}\ket{C_0}\nonumber\\
    +\sum_{l=1}^{d_E} c_l^1 e^{-\ii t\omega \lambda_1} \ket{\psi_l^1}\ket{C_1}.
\end{align*}
For the second term, we obtain
\begin{align*}
    \ket{\Psi_{EPA}(t)}&=e^{-\ii tPH_{EP}P}e^{-\ii tP\omega\Tilde{G}P}\ket{\Psi_{EPA}}
    \\
    &=\sum_{k=1}^{d_E} c_k^0 e^{-\ii t\omega \lambda_0}e^{-\ii t\psi_k^0 }\ket{\psi_k^0}\ket{C_0} \nonumber\\
    &\;\;+\sum_{l=1}^{d_E} c_l^1 e^{-\ii t\omega\lambda_1} e^{-\ii t\psi_l^1}\ket{\psi_l^1}\ket{C_1},
\end{align*}
where $\psi_i^j$ is the eigenvalue of $PH_{EP}P$ associated with the eigenstate $\ket{\psi_i}\ket{C_i}$. We trace out the environment to obtain the reduced state of the probe and auxiliary system
\begin{equation*}
    \rho_{PA}=\begin{pmatrix}
        \frac{1}{2} & e^{-\ii\omega t(\lambda_0-\lambda_1)}\alpha(t) \\
        e^{\ii\omega t\lrp{\lambda_0-\lambda_1}}\alpha(t)^* & \frac{1}{2}
    \end{pmatrix}.
\end{equation*}
We obtain a 2 by 2 matrix as the state is contained in the two-dimensional code space, which is spanned by the states $\ket{C_0},\ket{C_1}$. We know that the diagonals are $\frac{1}{2}$ since the initial state had amplitudes of $\frac{1}{\sqrt{2}}$ in $\ket{C_0}$ and $\ket{C_1}$ which are both eigenstates of $P\Tilde{G}P$ and $PH_{EP}P$. Therefore, the probabilities of measurement in that basis do not change. The value of $\alpha$ is given by 
\begin{equation*}
    \alpha=\sum_{k=1}^{d_E}\sum_{l=1}^{d_E}c_k^0{c_l^1}^* e^{-\ii t(\psi_k^0-\psi_l^1)}\braket{\psi_l^1}{\psi_k^0}.
\end{equation*}
Let $\Delta\lambda=\lambda_0-\lambda_1$.
The QFI of a mixed state can be calculated using the following equation from \cite{liu_quantum_2019}
\begin{equation}\label{eq:QFI_formula}
    \mathcal{F}(\rho)=2\sum_{i,j}\frac{\abs{\bra{\gamma_i}\del_{\omega}\rho\ket{\gamma_j}}^2}{\gamma_i+\gamma_j},
\end{equation}
where $\gamma_i$ and $\ket{\gamma_i}$ are the eigenvalue and eigenvector pairs of $\rho$, and $\del_{\omega}\rho$ is the derivative of $\rho$ with respect to the parameter $\omega$. The sum is taken over $i,j$ pairs such that $\gamma_i+\gamma_j>0$.

One can show, by direct calculation, that the eigenvalues of $\rho_{PA}$ are given by
\begin{equation}
    \gamma_{\pm}=\frac{1}{2}\pm\abs{\alpha(t)}.
\end{equation}
Since the reduced density matrix has positive eigenvalues between zero and one, we can deduce that $\abs{\alpha(t)}\in\lrb{-\frac{1}{2},\frac{1}{2}}$. We rewrite $\alpha(t)$ in the polar form as $\alpha(t)=\abs{\alpha(t)}e^{-\ii\phi_{\alpha}(t)}$. We will drop the explicit $t$ dependence from the notation to reduce clutter. Further, the eigenvectors of $\rho_{PA}$ are given by
\begin{equation}
    \ket{\gamma_{\pm}}=\frac{1}{\sqrt{2}} \begin{pmatrix}
        1 \\ \pm e^{\ii\omega t\Delta\lambda-\ii\phi_{\alpha}}
    \end{pmatrix}.
\end{equation}
Inserting these values into \cref{eq:QFI_formula} gives
\begin{equation*}
    \mathcal{F}(\rho_{PA})=4 t^2\lrp{\Delta\lambda}^2\abs{\alpha(t)}^2.
\end{equation*}
To show the HL, we need to show that 
    \begin{equation*}
 \limsup_{t\rightarrow\infty}\abs{\alpha}=\tx{constant}.
    \end{equation*}

    We already know that $\limsup_{t\rightarrow\infty}\abs{\alpha}\leq \frac{1}{2}$ from the earlier derived condition. We can rewrite $\alpha$ by letting $\bar{c}_m=c_k^0{c_l^1}^{*}\bra{\psi_l}\ket{\psi_k}$ 
    and $\phi_m=\psi_k^0-\psi_l^1$, where $m=(k,l)$ runs over the indices of $l$ and $k$. The equation for $\alpha(t)$ becomes
    \begin{equation}\label{eq:alpha_def}
        \alpha(t)=\sum_m\bar{c}_me^{-\ii t\phi_m},
    \end{equation}
     where $m$ runs over $d_E^2$ different values. It is important to note that the $\bar{c}_m$ do not depend on $t$.
    For $t=0$ we have $\rho_{PA}=\ket{\psi_{PA}}\bra{\psi_{PA}}$, and so $\alpha(0)=\frac{1}{2}$.
    Our goal is to prove 
    \begin{equation}
    \label{eq:limsup}
        \limsup_{t\rightarrow\infty}\abs{\alpha(t)}=\frac{1}{2}.
    \end{equation}
    This condition is very similar to quantum state revival in a closed quantum system, which is guaranteed by the Poincarr\'e recurrence theorem \cite{bocchieri_quantum_1957,schulman_note_1978,percival_almost_1961}. One can prove it by factoring the time evolution to separate signal and noise, and considering the recurrence times of the noise Hamiltonian. In fact, the revival of the QFI is directly implied from the theory of almost periodicity~\cite{percival_almost_1961,corduneanu_almost_2009}, which we argue below. 
    For completeness, we further provide an explicit proof of \cref{eq:limsup} in Appendix~\ref{section:direct_alpha_proof}.

    Remark $\alpha(t)$ is a uniformly convergent Fourier series as it contains finite $d_E^2$ terms, indicating its almost periodicity~\cite{percival_almost_1961,schulman_note_1978,corduneanu_almost_2009}. Then $\alpha(t)$, as an almost periodic function, satisfies $\forall\ \epsilon>0$ and $\forall\ t_0$, there exists an infinite many values of $\tau$ such that :
    \begin{equation}
        \abs{\alpha(t_0+\tau)-\alpha(t_0)}<\epsilon.
    \end{equation}
    The HL is achievable so long as we measure the system at proper times where $\alpha(t)$ is above a certain positive constant.
    Further, these values of $\tau$ are well spread out between $-\infty$ and $+\infty$ such that one cannot construct arbitrarily long intervals between two values. The almost periodicity implies that if we pick a random $t$ from a sufficiently large interval, there is a non-zero probability that $\abs{\alpha(t)}$ is above a positive constant. Therefore, even if we have no knowledge of the environmental state, we can still find a proper $t$ that achieves the HL by picking randomly. 
    \end{proof}


    Similar to Section~\ref{subsection:HNES}, the above protocol achieves $\limsup_{T\rightarrow 0}\mathcal F/T^2 > 0$ taking $T = t$ and $N = 1$ in \cref{fig:time_scales}, where QEC is performed fast and frequently and no intermediate measurements on the probe system are performed.

    As mentioned previously, measuring at random times will yield HL scaling; however, the constant factors in the scaling will likely be suboptimal. In practice, it is very difficult to compute the exact revival times for a system, even with knowledge of the dynamics. Instead, one can probe the revival times experimentally to calibrate the sensing protocol. This can be done with a Ramsey experiment, for example. If the revival times are too long, then it may be better to consider the HNES conditions and eliminate all interactions, if possible.

    It is worth noting that for systems with small environment dimensions and simple couplings, such as the previously discussed superconducting qubits interacting with two-level systems and nitrogen vacancy centers interacting with nearby spins, the envelope functions are relatively simple. Further, the revivals have a short period \cite{neeley_process_2008,shalibo_lifetime_2010,childress_coherent_2006,lu_observing_2020}, making them easy to probe.

\subsection{Diagonal Interactions and the HNELS Conditions}\label{subsection:HNELS}

We begin by defining the diagonal interaction. This is a special case of the model where the components of operators supported on the environment degrees of freedom all commute pairwise. 

\begin{definition}[Diagonal Interaction]
\label{def:diagonal_interaction}
We call a HMM in \cref{eq:master_eq} a HMM with diagonal interactions if and only if there exists an orthonormal basis $\{\ket{\phi_i}\}_{i=1}^{d_E}$ in $\mH_E$ such that 
\begin{gather}
    H_E=\sum_{i}h_{E}^i \ket{\phi_i}\bra{\phi_i}\otimes\ident_{PA}, \\
    H_{EP} = \sum_i \ket{\phi_i}\bra{\phi_i} \otimes H^{(i,i)}_{EP},\\ L_k = \sum_i \ket{\phi_i}\bra{\phi_i} \otimes L^{(i,i)}_{k},\quad \forall k.
\end{gather}
\end{definition}

For HMMs with diagonal interactions, the quantum dynamics can be simplified through the following lemma. 
\begin{lemma}\label{lemma:diagonal}
Consider a quantum state $\rho$ as an input state to an HMM with diagonal interactions. The output state after evolution of time $t$ is $\Phi_t(\rho)$ where $\Phi_t(\cdot)$ denotes the evolution channel after time $t$. We have 
\begin{equation}
\label{eq:diagonal-lemma}
    \trace_E(\Phi_t(\rho)) = \trace_E(\Phi_t((\mD_E \otimes \ident_{PA})(\rho)),
\end{equation}
where $\mD_E$ is the completely dephasing channel on $\mH_E$ such that $\mD_E(\cdot) = \sum_{k} \ket{\phi_k}\bra{\phi_k}(\cdot)\ket{\phi_k}\bra{\phi_k}$. 
In addition, \cref{eq:diagonal-lemma} still holds when quantum controls are applied on the probe and auxiliary system during the evolution. 
\end{lemma}
\begin{proof}
    To see \cref{eq:diagonal-lemma} holds, we observe it is sufficient to prove the commuting relation between $\mD_E$ and $\Phi_t$, because 
\begin{equation}
\trace_E(\Phi_t(\rho)) = \trace_E(( \mD_E \otimes \ident_{PA})(\Phi_t(\rho))). 
\end{equation}
The fact that $\mD_E \otimes \ident_{PA}$ commutes with $\Phi_t$ can be proven by noticing $\mD_E \otimes \ident_{PA}$ commutes with the differential in \cref{eq:master_eq} for diagonal interactions. Finally, $\mD_E \otimes \ident_{PA}$ also commutes with any quantum operations acting on the probe and auxiliary system because they act on different systems. 
\end{proof}

Lemma~\ref{lemma:diagonal} implies that for diagonal interactions, it is always sufficient to consider input states that are diagonal in the environmental basis $\{\ket{\phi_i}\}_{i=1}^d$ and the off-diagonal terms will not contribute to the final measurement outcomes at the probe and auxiliary system. Below, we will assume the quantum states under consideration $\rho \in \mH_{E}\otimes \mH_{P} \otimes \mH_{A}$ are diagonal in $\mH_{E}$, i.e., 
\begin{equation}
\label{eq:rho-diag}
    \rho = \sum_{i} \alpha_i \ket{\phi_i}\bra{\phi_i} \otimes \rho^{(i,i)},
\end{equation}
where $\rho^{(i,i)}$ are density operators on $\mH_{PA}$ and $\alpha_i$ are non-negative constants summing up to one. As we are considering HMMs with diagonal interactions, only $\rho^{(i,i)}$ will evolve over time under the master equation, and it is sufficient to track the evolution of $\rho^{(i,i)}$.

Further, when studying diagonal interactions, we can ignore the contribution from $\tilde{H}_E$ as it will not make a difference to the QFI.

\begin{lemma}\label{lemma:ignore_HE}
    Consider a quantum state $\rho$ in the form of~\cref{eq:rho-diag} as the input state to an HMM with diagonal interactions. We have that 
    \begin{align}
        \tr_E\lrp{\Phi_t(\rho)}=\tr_E\lrp{\hat{\Phi}_t(\rho)},
    \end{align}
    where $\hat{\Phi}_t(\cdot)$ denotes the evolution channel after time $t$ where we ignore the contributions of $\tilde{H}_E$. In other words, it is the time evolution under the equation
    \begin{multline}\label{eq:master_eq_no_HE}
        \frac{d\rho}{dt}=-\ii\lrb{H_{EP},\rho}-\ii\lrb{\omega \Tilde{G},\rho} \\ +\sum_{k=1}^r\lrp{L_k\rho L_k^{\dagger} -\frac{1}{2}\lrcb{L_k^{\dagger}L_k,\rho}}.
    \end{multline}
\end{lemma}

\begin{proof}
    First, notice that $\tilde{H}_E$ commutes with all other operators of the master equation. This allows us to factor the evolution as follows
    \begin{align*}
        \Phi_t(\cdot)=\hat{\Phi}_t\lrp{e^{-it\tilde{H}_E}(\cdot) e^{it\tilde{H}_E}}=e^{-it\tilde{H}_E}\hat{\Phi}_t(\cdot)e^{-it\tilde{H}_E}.
    \end{align*}
For a $\rho$ in the form of~\cref{eq:rho-diag}, time evolution from $H_E$ has no effect.
\begin{align*}
    e^{-it\tilde{H}_E}\rho e^{it\tilde{H}_E}&=\sum_i\alpha_i e^{-itH_E}\ket{\phi_i}\bra{\phi_i} e^{itH_E} \otimes \rho_{PA} \\
    &= \sum_i\alpha_i e^{-ith_E^i}\ket{\phi_i}\bra{\phi_i} e^{ith_E^i} \otimes \rho_{PA}\\
    &=\sum_i\alpha_i \ket{\phi_i}\bra{\phi_i} \otimes \rho_{PA}=\rho.
\end{align*}
And so, we see that the evolution due to $\tilde{H}_E$ has no effect, and only the evolution due to $\hat{\Phi}_t$ must be considered.
\end{proof}

When the interaction is diagonal, with less restrictive conditions than the HNES conditions, we can guarantee that the QFI scales as a constant times $t^2$ multiplied by some periodic envelope function, and that $\limsup_{t\rightarrow\infty}\mathcal{F}/t^2>0$ when $d_E<\infty$. This is the same scaling as the case of unitary evolution seen in ~\ref{subsection:unitary} and guaranteed by Theorem~\ref{theorem:unitary_scaling}.

\begin{theorem}[Hamiltonian not in extended Lindblad span (HNELS) Conditions]\label{theorem:HNELS} 
Consider the case of a diagonal interaction. Let $\lrcb{\ket{\phi_i}}_{i=1}^{d_E}$ be the orthonormal basis for $\Hammy_E$ from Definition~\ref{def:diagonal_interaction}. Let $\Hammy_E$ have finite dimension, i.e $d_E<\infty$. Let ${\mathcal{S}^{\circ}}$ be the linear span of the operators $\ident_P$, $\bra{\phi_i}L_k\ket{\phi_i},$ $\bra{\phi_i}L_k^{\dagger}\ket{\phi_i}$, and $\bra{\phi_i}L_k^{\dagger}\ket{\phi_i}\bra{\phi_m}L_j\ket{\phi_m}$  ($\forall\ k,j\in\lrcb{1,2,...,r}$, $\forall i,m\in\lrcb{1,2,...,d_E}$). ${\mathcal{S}^{\circ}}$ is similar to $\mathcal{S}$, except we exclude the contributions from $H_{EP}$ and only consider diagonal elements.

If $G\notin{\mathcal{S}^{\circ}}$, then $\limsup_{t\rightarrow\infty}\mathcal{F}/t^2>0$. 
Further, it is possible to achieve a QFI that scales as $t^2$ multiplied by a periodic envelope function. 
We will refer to the containment $G\notin{\mathcal{S}^{\circ}}$ as the ``Hamiltonian not in extended Lindblad span" (HNELS) conditions. 

\end{theorem}

\begin{proof}
We consider $\rho$ in diagonal form (\cref{eq:rho-diag}) since we trace out the environment at the end and may thus apply Lemma~\ref{lemma:diagonal} and Lemma~\ref{lemma:ignore_HE}. For a $\rho$ of this form, the master equation governing the HMM takes the form 
\begin{align*}
    &\frac{d\rho}{dt}=\\
    &\sum_i\alpha_i \ket{\phi_i}\bra{\phi_i}\otimes\left(-\ii\omega\lrb{ G,\rho^{(i,i)}}-\ii\lrb{H_{EP}^{(i,i)},\rho^{(i,i)}} + \right. \\
    &\left. \sum_{j=1}^r\lrp{L_j^{(i,i)}\rho^{(i,i)}\lrp{L_j^{(i,i)}}^{\dagger} - \frac{1}{2}\lrcb{\lrp{L_k^{(i,i)}}^{\dagger}L_k^{(i,i)},\rho^{(i,i)}} }\right)
\end{align*}

We employ the same construction for a QEC code as in Section~\ref{subsection:HNES}, except we swap $\mathcal{S}$ for $\mathcal{S}^{\circ}$.
For an initial state in the code space, the resulting dynamics are governed by the master equation
\begin{align}\label{eq:P_Lindblad_hnels}
    \frac{d\rho}{dt} &= \lim_{\Delta t\rightarrow 0}\frac{P \Phi_{\Delta t}(\rho) P + \mR(P_\perp \Phi_{\Delta t}(\rho) P_\perp)}{\Delta t} \nonumber\\
    &=\sum_i\alpha_i \ket{\phi_i}\bra{\phi_i}\otimes\left(-\ii\omega\lrb{PGP,\rho^{(i,i)}} \right. \nonumber\\
    &-\ii\lrb{PH_{EP}^{(i,i)}P,\rho^{(i,i)}} + \sum_{j=1}^r\left( PL_j^{(i,i)}\rho^{(i,i)}\lrp{L_j^{(i,i)}}^{\dagger}P \right. \nonumber\\
    &+ \mathcal{R}\lrp{P_{\perp}L_j^{(i,i)}\rho^{(i,i)}\lrp{L_j^{(i,i)}}^{\dagger}P_{\perp}} \nonumber\\
    &\left. \left.- \frac{1}{2}\lrcb{P\lrp{L_k^{(i,i)}}^{\dagger}L_k^{(i,i)}P,\rho^{(i,i)}} \right)\right).
\end{align}
where $P$ is the projection operator onto the code space. Since $G\notin\mathcal{S}^{\circ}$, we have that $PGP$ acts non-trivially in the code-space. Further, by design of the QEC code, the last three terms of \cref{eq:P_Lindblad_hnels} give $\rho^{(i,i)}$ up to a scalar. This scalar must be zero since the jump operator portion of the master equation has trace zero, and the channel $P\cdot P+\mathcal{R}(P\cdot P)$ is trace preserving. Therefore, the evolution is given by
\begin{align*}
    \frac{d\rho}{dt}=\sum_i\alpha_i\ket{\phi_i}\bra{\phi_i}\otimes -\ii\lrb{P(\omega G+H_{EP}^{(i,i)})P,\rho^{(i,i)}},
\end{align*}
which can be rewritten as
\begin{align}\label{eq:HNELS_unitary}
    \frac{d\rho}{dt}=-\ii\lrb{P\lrp{\omega\tilde{G}+\sum_i\ket{\phi_i}\bra{\phi_i}\otimes H_{EP}^{(i,i)}}P,\rho}.
\end{align}
We see from \cref{eq:HNELS_unitary} that after quantum error correction, for a diagonal interaction, the overall evolution of the system and environment is unitary. Applying Theorem~\ref{theorem:unitary_scaling} therefore yields the desired QFI scaling.

Note that in this proof, a first code was used to obtain unitary dynamics on the full system. In order to prove the scaling of Theorem~\ref{theorem:unitary_scaling}, the probe and auxiliary system were encoded in a QECC that imposed commutation between the signal and coupling Hamiltonian. It is not necessary here to concatenate the two codes, as the first code used in the proof already has all the properties necessary to prove the HL scaling.
\end{proof}

\subsection{The HNLS Conditions for when the HL is Unattainable}\label{subsection:HNLS}

Using the ``Hamiltonian not in Lindblad span" (HNLS) conditions of \cite{zhou_achieving_2018} on the full system, we can obtain conditions for when the HL is unattainable.

\begin{theorem}[Unattainability of the HL by violating the HNLS conditions~\cite{zhou_achieving_2018}]\label{theorem:HNLS_full_system}
    Let $\mathcal{S}_{EP}$ be the linear span of $\ident_{EP}$, $L_k$, $L_k^{\dagger}$, $L_k^{\dagger}L_j$ $\forall\ k,\ j$. If $\lrp{\ident_E\otimes G+H_E\otimes\ident_P} \in\mathcal{S}_{EP}$, then the SQL is the best achievable scaling. In other words, $\mathcal{F}= O(T)$.
\end{theorem}

The proof is as follows.
The HNLS conditions of \cite{zhou_achieving_2018} tell us that if $\ident_E\otimes G+H_E\otimes \ident_P \in\mathcal{S}_{EP}$, then the SQL is the best achievable asymptotic scaling with controls that act on the full system, i.e., probe, auxiliary system, and environment. In our analysis, we restricted ourselves to control operations acting only on the probe and auxiliary system, which is a strict subset of the set of full system controls. Therefore, the SQL is necessarily the best achievable asymptotic scaling for the local controls.

Further, we can relate these conditions to the previously derived HNES conditions (Theorem~\ref{theorem:HNES}). Consider

\begin{align*}
    \ident_E\otimes G+H_E\otimes\ident_P\in\mathcal{S}_{EP}\ \nonumber\\ \Rightarrow\  \ident_E\otimes G+H_E\otimes\ident_P =\sum_{A\in\mathcal{S}_{EP}}c_A A,
\end{align*}
where $c_A\in\C$. We simply trace out the environment to obtain the following:
\begin{equation*}
    G=\frac{1}{d_E}\sum_i\sum_{A\in \mathcal{S}_{EP}}c_A\bra{\phi_i}A\ket{\phi_i}-\tr\lrp{H_E}\ident_P.
\end{equation*}
From this equation, we see that $\lrp{\ident_E\otimes G+H_E\otimes\ident_P}\in\mathcal{S}_{EP}$ implies $G\in\mathcal{S}$.
We can obtain a simple yet useful result by taking the contrapositive of this statement. If $G\notin\mathcal{S}$ then $\lrp{\ident_E\otimes G+H_E\otimes\ident_P}\notin\mathcal{S}_{EP}$.
These containment conditions and their relations to the scaling of the QFI are presented in \cref{fig:venn_diagram}.

\subsection{Dependence of the Quantum Fisher Information on the Initial Environment State}\label{subsection:e_dependence}

We have so far shown that when the HNES conditions are satisfied ($G\notin{\mathcal{S}}$), the HL is achievable using QEC. We have further shown that $\lrp{\ident_E\otimes G+H_E\otimes\ident_P}\in\mathcal{S}_{EP}$ implies the HL is not attainable, using the HNLS conditions. This leaves one scenario to study, the case when $G\in{\mathcal{S}}$ but $\lrp{\ident_E\otimes G+H_E\otimes\ident_P}\notin \mathcal{S}_{EP}$.

When $G\in{\mathcal{S}}$ but $\lrp{\ident_E\otimes G+H_E\otimes\ident_P}\notin \mathcal{S}_{EP}$, the master equation (\cref{eq:master_eq}) does not fully determine the scaling of the QFI. We must also take into account the initial state of the environment and the trajectory of the evolution to find the scaling. We show this with the following example, for which we can be limited by either HL or SQL, depending on the initial environment state.

We wish to sense the parameter $\omega$ with signal Hamiltonian $\Tilde{G}=\ident_E\otimes Z$. We consider environment and probe, which are both qubits, so that $\Hammy_E\otimes\Hammy_P=\C^2\otimes\C^2$. Further, we assume no unitary coupling between the qubits, i.e. $H_{EP}=0$, and no signal on the environment $H_E=0$. We consider a model with the following two jump operators 
\begin{equation}
    L_1=\ket{0}\bra{0}\otimes\ident,\quad L_2=\ket{1}\bra{1}\otimes Z.
\end{equation}

Here, $Z$ refers to the Pauli $\sigma_Z$ matrix. We first consider the situation where the environment is initially in the $\ket{0}$ state, such that the initial density matrix is given by $\rho_{EP}(0)=\ket{0}\bra{0}\otimes \rho_{P,0}$ where $\rho_{P,0}$ is the initial probe state. 
One can see that throughout the evolution, the state will always remain in a product state. The environment remains unchanged and stays in the $\ket{0}\bra{0}$ state. We can therefore trace out the environment and look only at the dynamics on the probe. The reduced density matrix on the probe evolves under the equation
\begin{equation}
        \frac{d\rho_{PA}}{dt}=-\ii\omega\lrb{{G},\rho_{PA}}.
\end{equation}
For this choice of initial state, the probe evolution is entirely generated by the signal Hamiltonian, and the evolution is Markovian and unitary. This guarantees the achievability of the HL.

Similarly, consider the case where the initial environment state is $\ket{1}$. Once again, the state always remains in a product state throughout the evolution, and in particular, the environment always remains in the $\ket{1}\bra{1}$ state and does not change. We can trace out the environment and look at the evolution of the probe. In this case, the probe evolves under the following master equation
\begin{equation}
        \frac{d\rho_{PA}}{dt}=-\ii\omega\lrb{{G},\rho_{PA}}+Z\rho_{PA} Z-\rho_{PA}.
\end{equation}
The reduced dynamics are Markovian and therefore the results of \cite{zhou_achieving_2018} can be applied. For these dynamics, the HL cannot be achieved using any quantum controls due to the violation of the HNLS conditions. In this case, the SQL is the best achievable scaling since the signal Hamiltonian and the jump operator are both given by the same matrix.

\section{Achieving the SQL}\label{section:achieving_SQL}

We now consider the attainability of the SQL in the case where the HNES conditions are violated. If the HNES conditions are violated but the HNLS conditions on the full system are satisfied, then, as seen in Section~\ref{subsection:e_dependence}, the SQL can bound the QFI from above for certain initial states of the environment. Further, if the HNLS conditions aren't satisfied, then we are also guaranteed an SQL upper bound (Section~\ref{subsection:HNLS}). 
The question remains of whether linear scaling in $N$ is always achievable, in other words, if these bounds can always be saturated.

For the Markovian models considered in \cite{zhou_achieving_2018}, the SQL can be achieved using a prepare-and-measure strategy. It consists of preparing the probe in a state $\ket{\psi}$, letting it evolve for a constant time $t$ without applying any control operations, and then either measuring the state or shelving the state to replace it with a fresh probe. For this strategy in the Markovian regime, the system is completely reset after each measurement, and the SQL in $N$ is achievable due to the additivity of the QFI for product states. For the non-Markovian model, measurements on the probe do not reset the environment. This can lead to correlations between different probes, mediated by the environment, which in general can lead to higher or lower Fisher information for the distribution of measurement outcomes~\cite{radaelli_fisher_2023}.

In this section, we investigate whether linear scaling in $N$ is always achievable for the QFI. 
We first investigate this claim with an exactly solvable model, and then treat the general case with numerical methods. We numerically compute the FI for a finite number of probes and then linearly extrapolate  SQL scaling. 
The special case where the probe and environment undergo unitary evolution has been considered in \cite{guta_equivalence_2015}, where the authors show that a prepare-and-measure strategy achieves the SQL in $N$. 
We conjecture that the SQL remains achievable with this strategy when jump operators are added to the equation, and when the signal acts only on the probe. We leave a proof of this claim for the general non-Markovian setting as an open problem.

\subsection{Achieving the SQL for a Diagonal Interaction}\label{subsection:diag_sql}

We will consider a diagonal interaction as defined by Definition~\ref{def:diagonal_interaction}, as we did when considering the HNELS conditions of Section~\ref{subsection:HNELS}.
We simplify the dynamics further by applying QEC on the probe and auxiliary system to impose commutation of interactions on the probe degrees of freedom as well, allowing us to directly compute a QFI lower bound in a basis where all operators are simultaneously diagonal. To obtain the QFI lower bound, we will measure the probe at each time step and calculate the classical Fisher information (FI) of the probability distribution for the measurement outcomes. As mentioned above, we will employ the prepare-and-measure strategy. The FI of the joint probability distribution is a lower bound for the QFI of the quantum state of all output probes \cite{paris_quantum_2009}. 
Consequently, observing SQL scaling in the FI ensures a QFI scaling no less than SQL. In our case, the joint probability distribution for the measurement outcomes is a mixture of binomial distributions. The FI of such a distribution scales linearly in $N$ asymptotically. We shall use these arguments to prove the following theorem.

\begin{theorem}\label{theorem:SQL_diagonal_achievable}
    Consider the case of a diagonal interaction (Definition~\ref{def:diagonal_interaction}). With an interaction of this form, SQL scaling of the QFI is always achievable in $N$ using a combination of a prepare-and-measure strategy and QEC, for suitable choices of $t$.
\end{theorem}

\begin{proof}

The proof is as follows. Here, we define the QEC code in the same manner as in the proof of Theorem~\ref{theorem:unitary_scaling}, with $\mathcal{S}$, $\ket{C_0}$, and $\ket{C_1}$ defined in the same way. Contrary to the unitary case, in this section, we will need to perform the recovery operation during the QEC protocol. Although the code here cannot correct any non-trivial errors, it is sufficient to transform the dynamics on the probe system to commuting operations, with the recovery operation being 
\begin{multline}
\label{eq:recovery-dephasing}
    \mR(\cdot) = \ket{C_0}\bra{C_0} \trace( \ident_P \otimes (P_0)_A (\cdot)) \\ + \ket{C_1}\bra{C_1} \trace( \ident_P \otimes (P_1)_A (\cdot)). 
\end{multline}
For an initial state in the code space, the dynamics are governed by the following equation
\begin{align*}\label{eq:P_Lindblad_sql}
    \frac{d\rho}{dt} &= \lim_{\Delta t\rightarrow 0}\frac{P \Phi_{\Delta t}(\rho) P + \mR(P_\perp \Phi_{\Delta t}(\rho) P_\perp)}{\Delta t} \\
    &= -\ii\lrb{PH_{EP}P,\rho}-\ii\lrb{\omega P\tilde{G}P,\rho} + \sum_{i=1}^r \Big(P L_i \rho L_i P 
    \\ &\quad + \mR(P_\perp L_i \rho L_i^\dagger P_\perp ) - \frac{1}{2}\{P L_i^\dagger L_i P,\rho\}\Big),
\end{align*}
where $P$ is the projection operator onto the code space. Applying Lemma~\ref{lemma:ignore_HE} allows us to ignore the signal on the environment.
Assuming an initial state in the form of \cref{eq:rho-diag} and using \cref{eq:recovery-dephasing}, we have 
\begin{multline*}
    \frac{d\rho^{(i,i)}}{dt} =  -\ii\lrb{\omega P(G\otimes \ident_A)P,\rho^{(i,i)}} \\
    -\ii\lrb{PH_{EP}^{(i,i)}P + \frac{\Im(a_i)}{2} \bar{Z} ,\rho^{(i,i)}}+ \frac{\Re(a_i)}{2}\left( \bar{Z}\rho^{(i,i)}\bar{Z} - \rho^{(i,i)}\right),
\end{multline*}
where the logical operator $\bar{Z} = \ket{C_0}\bra{C_0} - \ket{C_1}\bra{C_1}$ and 
\begin{multline}
    a_i = \sum_k \frac{1}{2}\big(\bra{C_0}\lrp{L_k^\dagger L_k}^{(i,i)}\ket{C_0} + \bra{C_1}\lrp{L_k^\dagger L_k}^{(i,i)}\ket{C_1}) \\- \bra{C_0}L_k^{(i,i)}\ket{C_0}\bra{C_0}(L_k^\dagger)^{(i,i)}\ket{C_0}\big).
\end{multline}
It is clear that $\rho^{(i,i)}$ are evolving under Pauli-Z rotations and Markovian dephasing noise, and the strengths of the Pauli-Z rotation and the dephasing noise are different for different $i$.

We consider the prepare-and-measure strategy. At each time step, we prepare the state $\ket{\psi}=\frac{1}{\sqrt{2}}\lrp{\ket{C_0}+\ket{C_1}}$. For the measurement basis, we choose 
\begin{equation*}
    \ket{\pm_{\ii}}=\frac{1}{\sqrt{2}}\lrp{\ket{C_0} \pm \ii\ket{
    C_1}},
\end{equation*}
which we associate to the outcomes $1$ and $-1$ respectively. Let $X_n$ denote the random variable associated with the measurement at the $n$-th time step. We will denote by $b_n$ the outcome of the measurement for the $n$-th step. Consider a given measurement sequence with outcomes $b_1,\hdots,b_N$, the probability of this event occurring is given by the joint distribution $P(X_N=b_N,\hdots,X_1=b_1)$.

We will first consider the time evolution between two arbitrarily chosen steps of the prepare-and-measure protocol. 
As in the proof of Theorem~\ref{theorem:HNELS}, we will factor the evolution into two steps. 
The first step will be evolution under the master equation 
\begin{multline*}
    \frac{d\rho^{(i,i)}}{dt} =  -\ii\lrb{PH_{EP}^{(i,i)}P + \frac{\Im(a_i)}{2} \bar{Z} ,\rho^{(i,i)}}\\+ \frac{\Re(a_i)}{2}\left( \bar{Z}\rho^{(i,i)}\bar{Z} - \rho^{(i,i)}\right),
\end{multline*}
which corresponds to the evolution due only to the noise.
We can integrate the equation to obtain the CPTP map $\Lambda_t$ corresponding to the time evolution for a duration $t$ :
\begin{align}\label{eq:kraus_decomp}
    \Lambda_t\lrp{\ket{\psi}\bra{\psi}}=A_0\ket{\psi}\bra{\psi}A_0^{\dagger}+A_1\ket{\psi}\bra{\psi}A_1^{\dagger},
\end{align}
where
\begin{align*}
    A_0=&\sqrt{\frac{1+\exp\lrp{-\Re\lrp{\Gamma_i}t}}{2}} \exp\lrp{-\frac{\ii}{2} \Im\lrp{\Gamma_i}t \bar{Z}}, 
\end{align*}
and where
\begin{align*}
    A_1=&\sqrt{\frac{1-\exp\lrp{-\Re\lrp{\Gamma_i}t}}{2}}\bar{Z}\exp\lrp{-\frac{\ii}{2} \Im\lrp{\Gamma_i}t \bar{Z}}.
\end{align*}
Here, $\Gamma_i=\ii\tr\lrp{\bar{Z}PH_{EP}^{(i,i)}P} + a_i$. 
We have taken our initial probe state to be $\ket{\psi}$. Further, we know that for $t=0$, $\Lambda_t$ leaves the state unchanged and acts as the identity channel, as $A_0(t=0)=\ident_{PA}$ and $A_1(t=0)=0$. After evolving under this channel,
the next step is the unitary evolution due to the signal. 
\begin{align*}
    e^{-i\omega t P\lrp{G\otimes\ident_A}P}\Lambda_t(\ket{\psi}\bra{\psi})e^{i\omega t P\lrp{G\otimes\ident_A}P} \\
    =\sum_{k=0}^1 e^{-i\omega t P\lrp{G\otimes\ident_A}P}A_k \ket{\psi}\bra{\psi} A_k^{\dagger}e^{i\omega t P\lrp{G\otimes\ident_A}P}.
\end{align*}
The final step is to measure in the chosen basis. The probability of obtaining an outcome of $\pm$1 is given by

\begin{align*}
    &P\lrp{X_n=b_n|\rho_E=\ket{\phi_i}\bra{\phi_i}}\\ 
    =&\bra{\pm_{\ii}}e^{-i\omega t P\lrp{G\otimes\ident_A}P}\Lambda_t(\ket{\psi}\bra{\psi})e^{i\omega t P\lrp{G\otimes\ident_A}P}\ket{\pm_{\ii}}.
\end{align*}
First, we see that the probability of obtaining $b_n=1$ is independent of $n$ due to the fact that for an initial environment state of $\ket{\phi_i}\bra{\phi_i}$, the environment remains invariant and does not mediate correlations between measurements.
Let $p_i:=P\lrp{X_n=1|\rho_E=\ket{\phi_i}\bra{\phi_i}}$ and $q_i=1-p_i$ be the probability of obtaining $b_n=-1$.
The probability of obtaining $+1$ is given by
\begin{align}
    &p_i=\nonumber \\
    &\frac{1}{2} +\frac{1}{2}\exp\lrp{-\Re(\Gamma_i)t}\sin\lrp{\lrp{\Im\lrp{\Gamma_i}+\omega\lrp{\lambda_0-\lambda_1} }t},\label{eq:pi_analytic}
\end{align}
where $\lambda_0,\lambda_1$ are the eigenvalues of $P(G\otimes\ident_A)P$ associated to $\ket{C_0},\ket{C_1}$. Further, $\lambda_0\neq\lambda_1$ as we assume $G$ acts non-trivially on the probe. The derivative of the probability with respect to the sensing parameter $\omega$ is given by
\begin{align}
    \del_{\omega}p_i=\frac{1}{2}\exp\lrp{-\Re\lrp{\Gamma_i}t}\lrp{\lambda_0-\lambda_1}t\nonumber\\
    \cos\lrp{\lrp{\Im\lrp{\Gamma_i}+\omega\lrp{\lambda_0-\lambda_1} }t}.\label{eq:diff_pi_analytic}
\end{align}
Since the measurement outcomes are independent, the joint probability can be expressed as the product of the marginals
\begin{align}\label{eq:joint_rhoi}
     &P\lrp{X_N=b_n,\hdots,X_1=b_1|\rho_E=\ket{\phi_i}\bra{\phi_i}}\nonumber \\
    = & \prod_{n=1}^N p_i^{(1+b_n)/2}(1-p_i)^{(1-b_n)/2}.
\end{align}

For an arbitrary diagonal environment state, such as in \cref{eq:rho-diag}, we sum the joint probabilities (\cref{eq:joint_rhoi}) for the different environment basis states 
\begin{align*}
    &P\lrp{X_N=b_n,\hdots,X_1=b_1}\\
    =&\sum_{i=1}^{d_E}\alpha_i\prod_{n=1}^N p_i^{(1+b_n)/2}(1-p_i)^{(1-b_n)/2}.
\end{align*}
As mentioned previously, the order of the outcomes doesn't matter, and we recognize a sum of binomial distributions, also called a binomial mixture. For a measurement sequence $b_1,\hdots,b_N$, let $N_1$ count the $+1$ outcomes in the sequence. The associated probability distribution for $N_1$ is 
\begin{align}\label{eq:binomial_mixture}
    P(N_1=n_1)=\sum_{i=1}^{d_E}\alpha_i\binom{N}{n_1}p_i^{n_1}(1-p_i)^{N-n_1}.
\end{align}
From this point on, we will assume without loss of generality that all the $p_i$ are distinct when we evaluate them at the true value of $\omega$, and where a specific value of $t$ has already been fixed. We may do this since similar values of $p$ can be combined into single terms of the sum. Suppose, for example, that $p_2=p_3$, then we can set $\alpha_2\mapsto \alpha_2+\alpha_3$, and remove the index 3 from the sum.

To show that this strategy achieves the SQL, one must show that the FI of this binomial mixture scales linearly in $N$ asymptotically. We show this in Appendix~\ref{section:binomial_linear_FI}.  

\end{proof}


\subsection{Two Qubit Dephasing Model}\label{subsection:two_qubit_dephasing}

We will now consider a model that has a diagonal interaction, but for which we do not need to apply QEC. Further, for this model, we will not measure or compute FI. Rather, we will shelve the probe states and replace them with new ones. For certain control parameters, the QFI of this state will be computed exactly to show SQL scaling.

We now consider a model where $\Tilde{G}$, $H_{EP}$, $\tilde{H}_E$, and $L_k$ all commute pairwise. We model both the environment and probe by qubits $\Hammy_P\simeq \Hammy_E\simeq \C^2$. Let $G=Z$ and $H_{EP}=Z\otimes Z$. Let there be a single jump operator $L=\ident_E\otimes Z$. Without loss of generality, we set $\tilde{H}_E=0$, as the commutation of the model allows for the application of Lemma~\ref{lemma:ignore_HE}. This model corresponds to pure phase noise on both the system and environment, as the interaction is completely diagonal on both the probe and environment. 

We can see that since $\Tilde{G}=L$, the conditions of Theorem~\ref{theorem:HNLS_full_system} are satisfied, and so the SQL is the best attainable scaling. Further, the results from Theorem~\ref{theorem:SQL_diagonal_achievable} of the previous section guarantee that the SQL is achievable with a prepare-and-measure strategy. This example, however, is particular since we can show that the QFI itself scales linearly in $N$, rather than a FI lower bound. Further, QEC will not be used in this example to tailor the noise.

We will prepare $N$ probe states in the $\ket{+}=\frac{1}{\sqrt{2}}\lrp{\ket{0}+\ket{1}}$. Since all the operators commute, we can factor out the time evolution into three steps: evolution by the channel corresponding to $L$, evolution $H_{EP}$, and finally, evolution by $G$.

Let $\rho_E$ be the initial state of the environment. We begin with the evolution from the jump operators, 
\begin{equation*}
    \frac{d\rho}{dt}=\lrp{\ident_E\otimes Z}\rho\lrp{\ident_E\otimes Z}-\rho.
\end{equation*}
One can show that evolving under this equation for a duration $t$ is equivalent to applying the following dephasing channel on the probe:
\begin{equation*}
    \mathcal{D}_t(\rho_{PA})=\lrp{\frac{1+e^{-2t}}{2}}\rho_{PA}+\lrp{\frac{1-e^{-2t}}{2}}Z\rho_{PA}Z.
\end{equation*}
For $\rho_{PA}=\ket{+}\bra{+}$, this channel is equivalent to a depolarizing channel :
\begin{equation*}
    \mathcal{D}_t\lrp{\ket{+}\bra{+}}=e^{-2t}\ket{+}\bra{+}+(1-e^{-2t})\frac{\ident_P}{2}.
\end{equation*}
Next, we apply the time evolution due to $H_{EP}$ to the state. The unitary operator for this time evolution is 
\begin{align*}
    U_{EP}(t) &= e^{-\ii tZ\otimes Z}\nonumber \\
    &=\ket{0}\bra{0}\otimes e^{-\ii tZ}+\ket{1}\bra{1}\otimes e^{\ii tZ}.
\end{align*}
Applying this unitary yields
\begin{align*}
    &U_{EP}(t)\left[\rho_E\otimes\mathcal{D}_t(\ket{+}\bra{+})\right]{U_{EP}(t)}^{\dagger}\\ 
    &=\sum_{i,j=0}^1 \ket{i}\bra{i}\rho_E\ket{j}\bra{j}\otimes e^{-\ii(-1)^itZ}\mathcal{D}_t(\ket{\psi}\bra{\psi})e^{\ii(-1)^jtZ}.
\end{align*}
Finally, we consider the time evolution due to the signal. This corresponds to unitary evolution on the probe generated by $G$. The state of the environment and the first probe after time $t$ is given by
\begin{align*}
    &\sum_{i,j=0}^1 \ket{i}\bra{i}\rho_E\ket{j}\bra{j}\otimes \nonumber \\ 
    &e^{-\ii\omega t Z}e^{-\ii(-1)^itZ}\mathcal{D}_t(\ket{\psi}\bra{\psi})e^{\ii(-1)^jtZ}e^{\ii\omega t Z}.
\end{align*}
Since we are considering a prepare-and-measure strategy, after this time $t$, the probe is swapped out for a new probe and the process is repeated. We repeat the process for $N$ probes. One can show that the final output state $\rho_{EP}^f$ for the environment and the $N$ probes is given by
\begin{align*}
    \rho_{EP}^f(t)= & \sum_{i,j=0}^1 \ket{i}\bra{i}\rho_E\ket{j}\bra{j}\otimes\nonumber \\ 
    &\lrb{e^{-\ii\omega t Z}e^{-\ii(-1)^itZ}\mathcal{D}_t(\ket{\psi}\bra{\psi})e^{\ii(-1)^jtZ}e^{\ii\omega t Z}}^{\otimes N}.
\end{align*}
Since we have access only to the states of the probe, we trace out the environment to obtain the final probe state $\rho_{PA}^f$,
\begin{align}
    \rho_{PA}^f(t)=& \sum_{i=0}^1\bra{i}\rho_E\ket{i}\left[ e^{-2t} e^{-\ii\omega t Z}e^{-\ii(-1)^itZ} \right.\nonumber \\ 
    &\left.\cdot \ket{+}\bra{+} e^{\ii(-1)^itZ}e^{\ii\omega tZ} + (1-e^{-2t})\frac{\ident_P}{2}\right]^{\otimes N}.
\end{align}
If we fix the time $t=\frac{\pi}{2}$, then this density matrix is easily diagonalizable. Indeed, the state reduces to
\begin{equation*}
    \rho_{PA}^f(\frac\pi2) = \lrb{e^{-\pi}e^{-\ii\omega \frac{\pi}{2} Z}\ket{-}\bra{-}e^{\ii\omega \frac{\pi}{2} Z}+(1-e^{-\pi})\frac{\ident_P}{2}}^{\otimes N}.
\end{equation*}
Since the QFI of a product state is the sum of the QFIs of the factors of the product \cite{liu_quantum_2019}, then the Fisher information of the final state $\mathcal{F}\lrp{\rho_{PA}^f\lrp{\frac{\pi}{2}}}$ is given by 
\begin{align}
\mathcal{F}\lrp{\rho_{PA}^f\lrp{\frac{\pi}{2}}}=N\mathcal{F}\lrp{\rho_1},
\end{align}
where 
\begin{align}
    \rho_1=e^{-\pi}e^{-\ii\omega \frac{\pi}{2} Z}\ket{-}\bra{-}e^{\ii\omega \frac{\pi}{2} Z}+(1-e^{-\pi})\frac{\ident_P}{2}.
\end{align}
So long as $\mathcal{F}(\rho_1)$ is independent of $N$ and greater than zero, then the SQL is achieved.

Recall that the QFI of a mixed state can be calculated using \cref{eq:QFI_formula}. One can show that $\rho_1$ has eigenvalues $\lrp{e^{-2t}+\lrp{1-e^{-2t}}/{2}}$ and $\lrp{1-e^{-2t}}/{2}$ with eigenvectors $e^{-\ii\omega\frac{\pi}{2}Z}\ket{-}$ and $e^{-\ii\omega\frac{\pi}{2}Z}\ket{+}$ respectively. This yields a QFI of
\begin{align}
    \mathcal{F}(\rho_1)=e^{-2\pi}\pi^2
\end{align}
which is both independent of $N$ and greater than zero. Therefore, for the diagonal interaction, the SQL is attainable with a prepare-and-measure strategy without using QEC to tailor the noise, like in the Markovian case.

\subsection{Numerical Sampling of the Classical Fisher Information}\label{subsection:sampling_FI}

In the previous subsections, we considered models with restricted structure, particularly diagonal interactions. We showed SQL scaling in those specific cases. We now consider more generic models, and present numerical evidence that the SQL is still generally achievable with prepare-and-measure strategies, even without QEC, mirroring the Markovian case \cite{zhou_achieving_2018,zhou_asymptotic_2021}.

Consider the joint quantum state of $N$ probes. If these probes are qudits, then the Hilbert space scales as $d_P^N$, and so the computation of the QFI using exact diagonalisation and \cref{eq:QFI_formula} quickly becomes intractable as $N$ grows. To simplify the problem, we instead measure the probes at each time step and calculate the FI.

For $N$ rounds of prepare-and-measure, we can associate the measurements with random variables $\{X_1,\cdots,X_N\}$, forming a stochastic process. These random variables can be correlated, with the correlations mediated by the environment. It has been shown that if the process is stationary and its Markov order, which quantifies the memory length, is finite, then the classical Fisher information will scale linearly in $N$ \cite{radaelli_fisher_2023}. It is the minimum $l$ such that
\begin{equation*}
    P\lrp{X_n|X_{n-1},..,X_1}=P\lrp{X_n|X_{n-1},...,X_{n-l}}.
\end{equation*}
The HMM under consideration does not generally possess these two properties.

We compute the classical Fisher information using \cref{eq:classical_FI}.
Since the same probe state is prepared at every time step, the $i$-th measurement outcome depends solely on the residual state of the environment after the previous measurement. Consequently, it suffices to track only the $d_E \times d_E$ density matrix of the environment and its derivative throughout the time evolution.
This approach enables the computation of the probabilities and their derivatives for the measurement outcomes $\{x_1, \cdots, x_N\}$ using constant memory, thereby avoiding the exponential memory requirements associated with storing the full density matrix for $N$ probes.
We trade the memory overhead for exponential time complexity. If each random variable $X_i$ takes on $n_\text{out}$ possible different values, then there are $(n_\text{out})^N$ different measurement outcomes, which must all be calculated separately to compute the FI exactly.

To get around this exponential scaling, we use random sampling to approximate the FI. 
Our methods are inspired by previous works using Monte Carlo to calculate the FI \cite{spall_monte_2005,spall_improved_2008}.
We consider the following unbiased estimator for the FI:
\begin{align}\label{eq:estimator_FI}
    \hat{\mathcal{I}}(\omega)=\frac{1}{S}\sum_{s=1}^S\lrp{\frac{\del_{\omega}P(x_1^{(s)},\cdots,x_N^{(s)})}{P(x_1^{(s)},\cdots,x_N^{(s)})}}^2,
\end{align}
where $\{x_1^{(s)},\cdots,x_N^{(s)}\}$ is a sample from $\{X_1,\cdots,X_N\}$ and the $s$ index differentiates the samples. The variable $S$ is the total number of samples. 
Each sample is obtained by computing the time evolution. At each measurement step, we randomly decide the outcome. To do so, we randomly sample from a Bernoulli distribution with parameter $p=P(X_n=x_n|X_{n-1}=x_{n-1},\hdots)$, where $P(X_n=x_n|X_{n-1}=x_{n-1},\hdots)$ is given by the trace of the non normalized environment density matrix after projecting the probe into the state associated with a $+1$ outcome.

This method provides an approximate value of the FI. 
For a fixed $N$, the variance of this estimator scales as $O(S^{-1})$, and so convergence is guaranteed for a large number of samples. The constant hidden by the big-$O$ may depend on $N$ and $t$. We a priori do not know this functional dependence. We can, however, numerically and heuristically estimate the variance to provide error bars for \cref{fig:heisenberg}. The details of the variance estimation are provided in Appendix~\ref{section:estimate_variance_estimator}.

\subsubsection{Heisenberg Interaction}

We first consider the prepare-and-measure strategy for a model where the operators of the master equation do not commute. We once again model the probes and environment as qubits $\Hammy_P\simeq \Hammy_E\simeq \C^2$. We consider a signal given by $G=Z$ and a single jump operator $L=\ident_E\otimes Z$. We consider a Heisenberg-type interaction for the coupling Hamiltonian 
\begin{align}
H_{EP}=\gamma\lrp{X\otimes X+Z\otimes Z+Y\otimes Y} 
\end{align}
where $\gamma$ is the coupling strength. We will consider two cases of the model. In the first case, there is no signal acting on the environment and $H_E=0$. In the second case, there is a signal acting on the environment and $H_E=Z$.

We fix here the measurement basis to be the eigenstates of the Pauli $X$ operator, and we prepare each probe in the state $\ket{+_{\ii}}=\frac{1}{\sqrt{2}}\lrp{\ket{0}+\ii\ket{1}}$. We fix the parameters to be $\gamma=1$, $\omega=0$, and $t=0.25$. The initial state of the environment is obtained by sampling pure states from the Haar measure. The same five random states are used for both the $H_E=0$ and $H_E=Z$ models. 
The Fisher information is then calculated using 25000 samples.

\cref{fig:heisenberg} shows the classical Fisher information as a function
of the number of prepare-and-measure rounds for the model where $H_E=0$.
The FI curves are calculated for five pure Haar random initial environment states $\{\rho_0,\cdots,\rho_4\}$. For $N\leq 15$, the FI is computed both exactly and by random sampling. 
At larger $N$, the values are exclusively from random sampling.

\cref{fig:heisenberg} a) shows the FI values up to $N=60$ as well
as linear regression curves for the data. For each
initial environment state $\rho_i$, a linear regression is performed on the
values of the FI to show the linear (SQL) scaling. The results of the linear
regressions are presented in Table \ref{tab:linregress}.
The linear regression gives us the slope of the linear relation, along with the
Pearson correlation coefficient $r_{\rm c}$, which measures the linear correlation
between $N$ and the FI. A higher $r_{\rm c}$ indicates a stronger correlation, with
$r_{\rm c} \in\lrb{-1,1}$. The $r_{\rm c}$ coefficients of Table \ref{tab:linregress} are all
very close to 1, which strongly indicates that the relationship between the
Fisher information and $N$ is linear. Further, the values for the slopes are all
above zero. Therefore, we see that the SQL is achieved up to depths of 60.
We conjecture that this trend holds
in the limit of $N\rightarrow\infty$.

\begin{table}[b]
  \caption{\label{tab:linregress}
    Linear regressions for the Fisher information as a function of the number of probes for the Heisenberg-type model with $H_E=0$. Associated curves are shown in \cref{fig:heisenberg}.}
  \begin{ruledtabular}
    \begin{tabular}{ccc}
      $i$ & Slope & $r_{\rm c}$ Coefficient \\
      \hline
      \hline
      0 & 0.05755739026318657 & 0.9987758048158253 \\
      1 & 0.057886505716965755 & 0.997690310855459 \\
      2 & 0.05267854285583423 & 0.9980036537841903 \\
      3 & 0.05734837817591377 & 0.9990203746173694 \\
      4 & 0.0643685351834392 & 0.9956504727114459 \\
    \end{tabular}
  \end{ruledtabular}
\end{table}

\cref{fig:heisenberg} b) shows the FI values up to $N=15$ from both
exact and random sampling methods. The curves in this figure serve only to guide the
eye. This figure shows the agreement between the sampling results and exact computations.

\cref{fig:heisenberg_HE} and \cref{tab:linregress_HE} show the same information but for the $H_E=Z$ model, where the signal also acts on the environment. Similarly to \cref{fig:heisenberg} b), \cref{fig:heisenberg_HE} b) shows excellent agreement between the exact and sampled Fisher information. In \cref{fig:heisenberg_HE} a), the presence of signal on the environment seems to lower the values of the FI. Further, the scaling of the FI does not seem to be linear. This is reflected by the values of $r_c$ obtained for the model in \cref{tab:linregress_HE}, which are much lower than those of \cref{tab:linregress}. In particular, the values of $r_c$ are not high enough to confidently conclude that the SQL is achieved.

\begin{table}[b]
  \caption{\label{tab:linregress_HE}
    Linear regressions for the Fisher information as a function of the number of probes for the Heisenberg-type model with $H_E=Z$. Associated curves are shown in \cref{fig:heisenberg_HE}.}
  \begin{ruledtabular}
    \begin{tabular}{ccc}
      $i$ & Slope & $r_{\rm c}$ Coefficient \\
      \hline
      \hline
      0 & 0.00893164136192775 & 0.937906503836237 \\
      1 & 0.008752382336452712 & 0.9550822417044361 \\
      2 & 0.006298895354174299 & 0.9968012417020177 \\
      3 & 0.00863620471038918 & 0.9498835201202465 \\
      4 & 0.007083260388947953 & 0.9857235070906991 \\
    \end{tabular}
  \end{ruledtabular}
\end{table}

\begin{figure}
    \centering
    \includegraphics[width=1\linewidth]{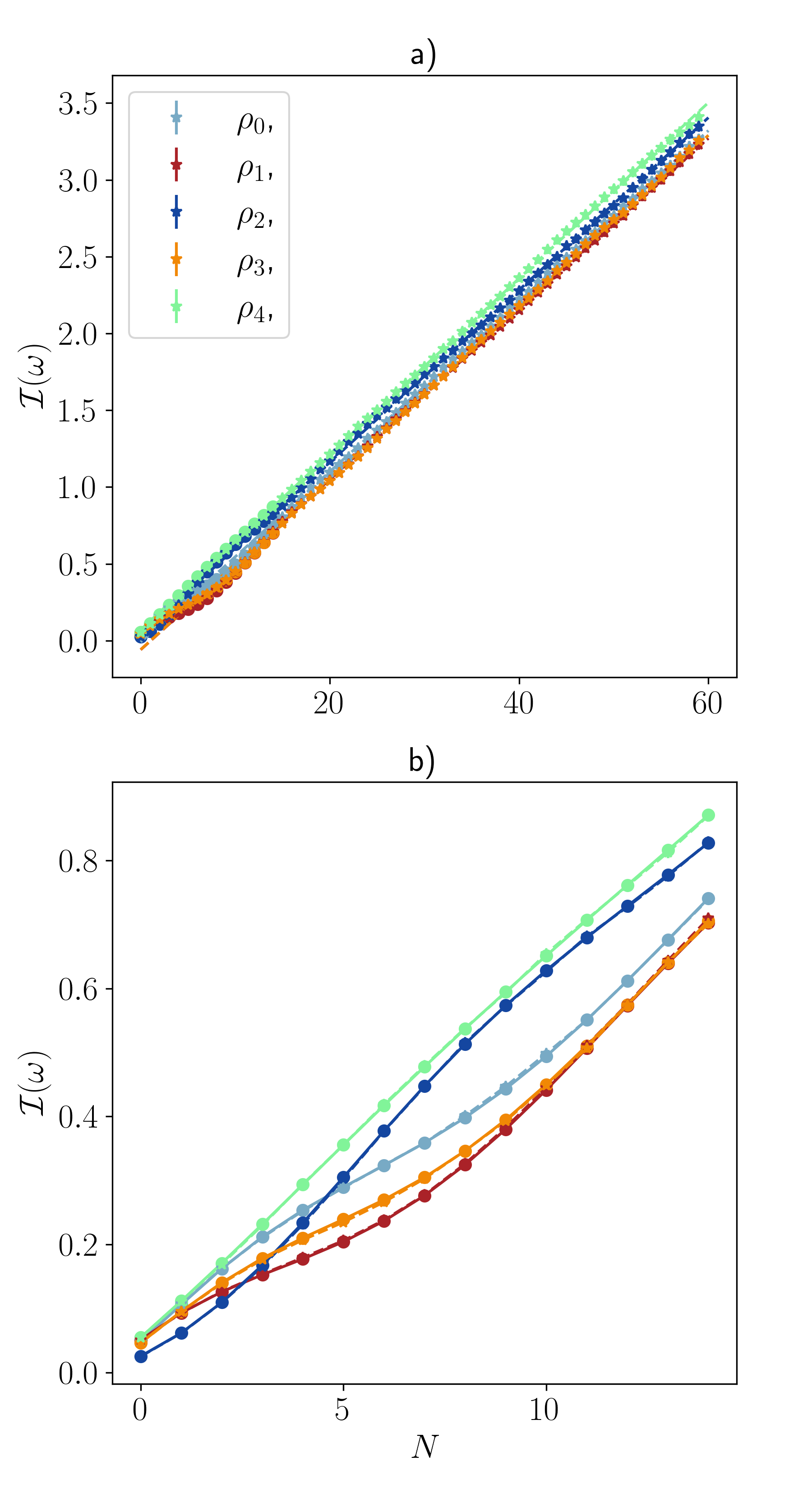}
    \caption{\label{fig:heisenberg}FI for the Heisenberg model with $H_E=0$. Solid dots were
    calculated exactly, and stars were calculated by random sampling. Different
    colours correspond to different initial environment states $\rho_i$ sampled
    from Haar random states. (a) FI for system sizes up to $N=60$. Dashed lines
    show linear regression on the data. (b) Comparison of exact computation and Monte Carlo sampling FI
    values up to $N=15$. Lines serve to guide the eye. Dashed and solid lines
    interpolate the random sampling and exact FI values, respectively. }
\end{figure}

\begin{figure}
    \centering
    \includegraphics[width=1\linewidth]{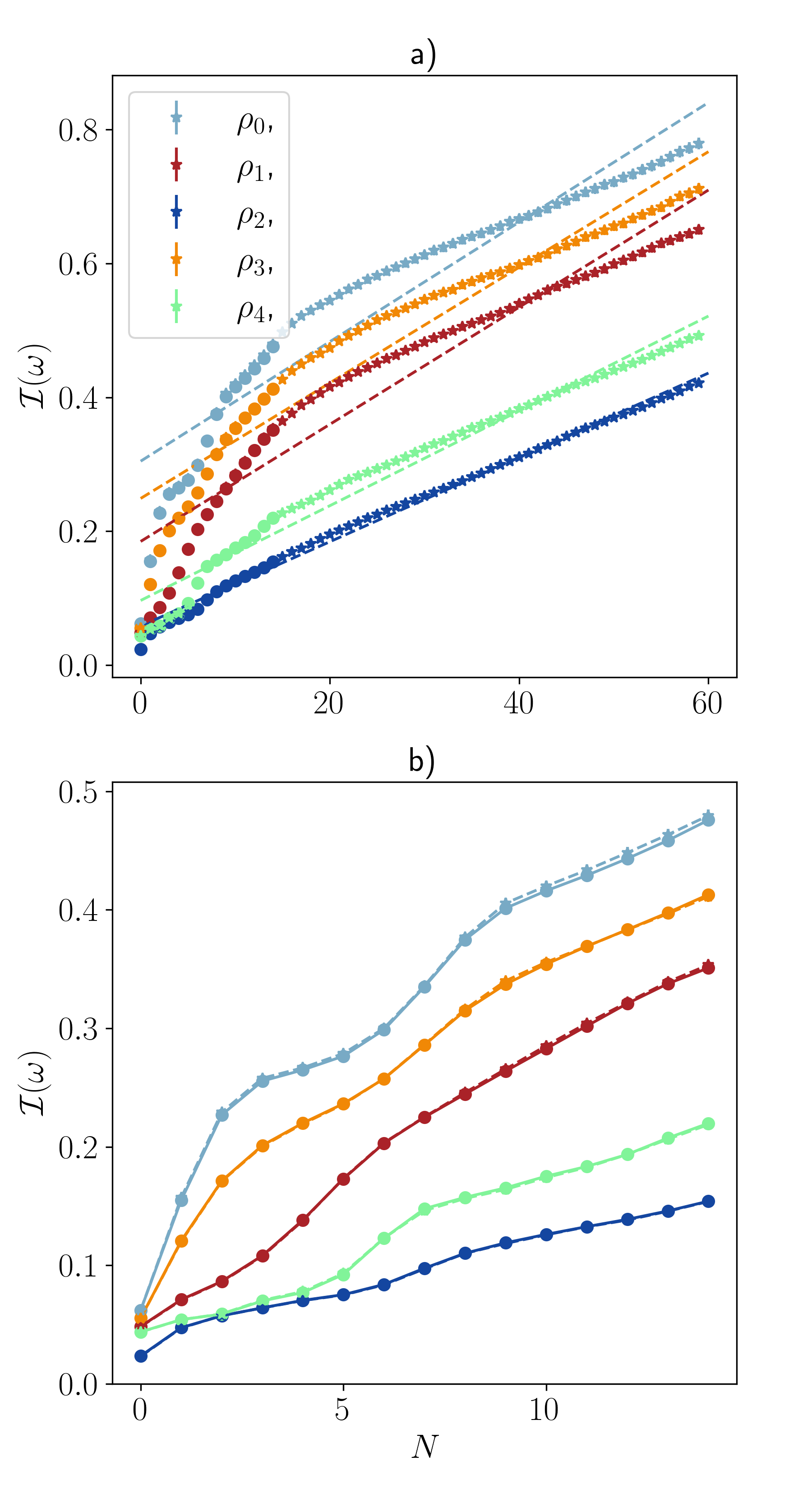}
    \caption{\label{fig:heisenberg_HE}FI for the Heisenberg model with $H_E=Z$. Solid dots were
    calculated exactly, and stars were calculated by random sampling. Different
    colours correspond to different initial environment states $\rho_i$ sampled
    from Haar random states. (a) FI for system sizes up to $N=60$. Dashed lines
    show linear regression on the data. (b) Comparison of exact computation and Monte Carlo sampling FI
    values up to $N=15$. Lines serve to guide the eye. Dashed and solid lines
    interpolate the random sampling and exact FI values, respectively. }
\end{figure}

\subsubsection{Random Master Equations}

To provide further evidence that the SQL is always achievable, we numerically study the prepare-and-measure strategy for a model where the noise operators of the master equation are randomly generated. We again model the probe and environment by qubits, $\Hammy_P\simeq\Hammy_E\simeq\C^2$. The protocol does not require any auxiliary degrees of freedom. We once again consider the signal to be $G=Z$ and suppose the parameter we are sensing is $\omega=0$. For each round of the protocol, we prepare the probe in the state $\ket{+_{\ii}}=\frac{1}{\sqrt
2}\lrp{\ket{0}+\ii\ket{1}}$, allow the system to evolve for a time $t=0.1$,
before finally measuring in the eigenbasis of the $X$ operator. As in the Heisenberg case, we consider two cases of the model, $H_E=0$ and $H_E=Z$.

We consider a model with three jump operators. 
The jump operators are obtained by sampling random matrices from the Ginibre ensemble \cite{ginibre_statistical_1965,can_random_2019}. For the coupling Hamiltonian $H_{EP}$, we sample a random matrix $A$ from the Ginibre ensemble. We then set $H_{EP}=\frac{1}{2}\lrp{A+A^{\dagger}}$. Finally, the initial environment state is taken to be a Haar random pure state.

For a given master equation, we compute the FI as a function of the number of rounds $N$ from $N=1$ to $N=30$. The values of the FI are calculated using the random sampling method of Section~\ref{subsection:sampling_FI} with 25000 samples. A linear fit is then applied to the resulting curve, yielding a slope and an $r_{\rm c}$ coefficient. These quantities were computed for 1050 different random master equations. The same 1050 master equations are used for the $H_E=0$ and $H_E=Z$ models.

The distributions of the $r_{\rm c}$ values and slopes for the $H_E=0$ models are presented in \cref{fig:boxes_no_HE} a) and \cref{fig:boxes_no_HE} b) respectively. 
Firstly, we see from \cref{fig:boxes_no_HE} a) that all of the obtained $r_{\rm c}$ values were
very close to 1, with the lowest being 0.9820 rounded to the fourth digit. There
is therefore a strong linear relationship between the FI and $N$ for all of the
random models studied. Therefore, it appears that linear scaling of the FI up
to small fluctuations is a general property with this choice of strategy.
Further, the minimum slope obtained was approximately 1.48466$\cdot 10^{-5}$. Although the
minimal slope is quite small, it is still nonzero and of a few orders of magnitude
above machine precision. All of the randomly generated models, therefore, have
approximately linear FI scaling with a nonzero slope. This is strong evidence
that the prepare-and-measure strategy can generically achieve the SQL. It is
important to note that since we chose the same values of $t$, the same
measurement basis, and the same initial state for all of the random noise models
considered, the values of the FI obtained may be far from optimal, even among
prepare-and-measure strategies. We can therefore expect an increase in
performance by fine-tuning these parameters for each random model considered.

\cref{fig:boxes_HE_Z} presents the same data for the $H_E=Z$ case of the random models. The minimum $r_c$ value among the models is 0.8058, rounded to the fourth digit, and the minimum slope is 3.2925e-06. We see that the majority of the $r_c$ values of~\cref{fig:boxes_HE_Z} a) are clustered between 0.99 and 1, indicating that for the majority of the random models, the relationship is strongly linear and therefore the scaling is probably SQL. However, there are three noticeable outliers. It is therefore unclear whether the SQL is always achieved in the presence of an environment signal, $H_E=Z$.

\begin{figure}
    \centering
    \includegraphics[width=1\linewidth]{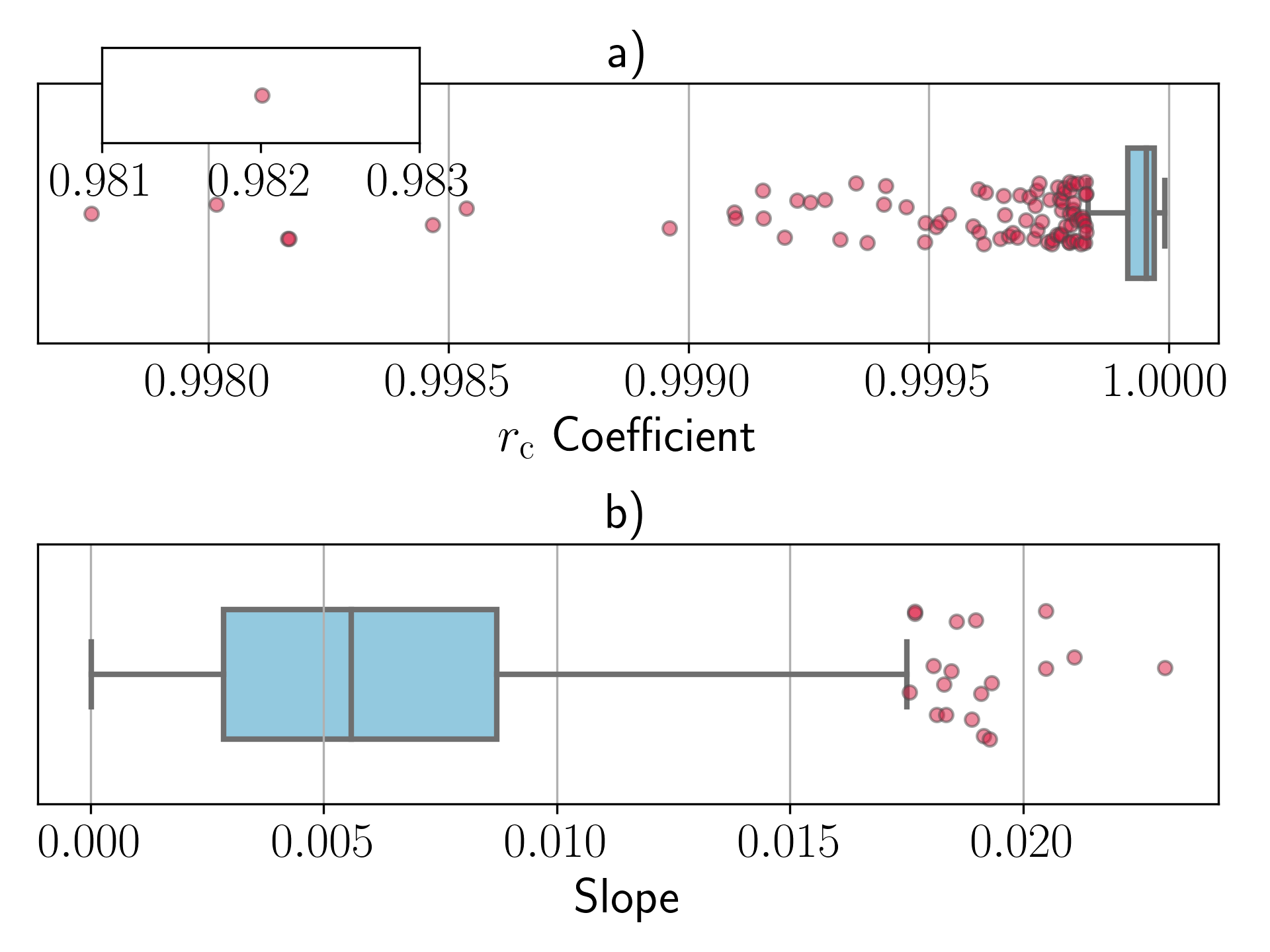}
    \caption{\label{fig:boxes_no_HE}Statistics for the linear fits to the FI of 1050
    random models with $H_E=0$. The FI of each data point is obtained with a prepare and
    measure strategy with noise from a random master equation. 
    The boxes extend from the first to the third quartile, split by the median. The whiskers extend from the quartiles to 1.5 times the inter-quartile range. 
    (a) Slopes of the linear fits. The sub-plot in the top-left corner shows an isolated outlier and has been separated for aesthetic purposes. (b) Pearson coefficients $r_{\rm c}$ of the linear fits.  
    }
\end{figure}

\begin{figure}
    \centering
    \includegraphics[width=1\linewidth]{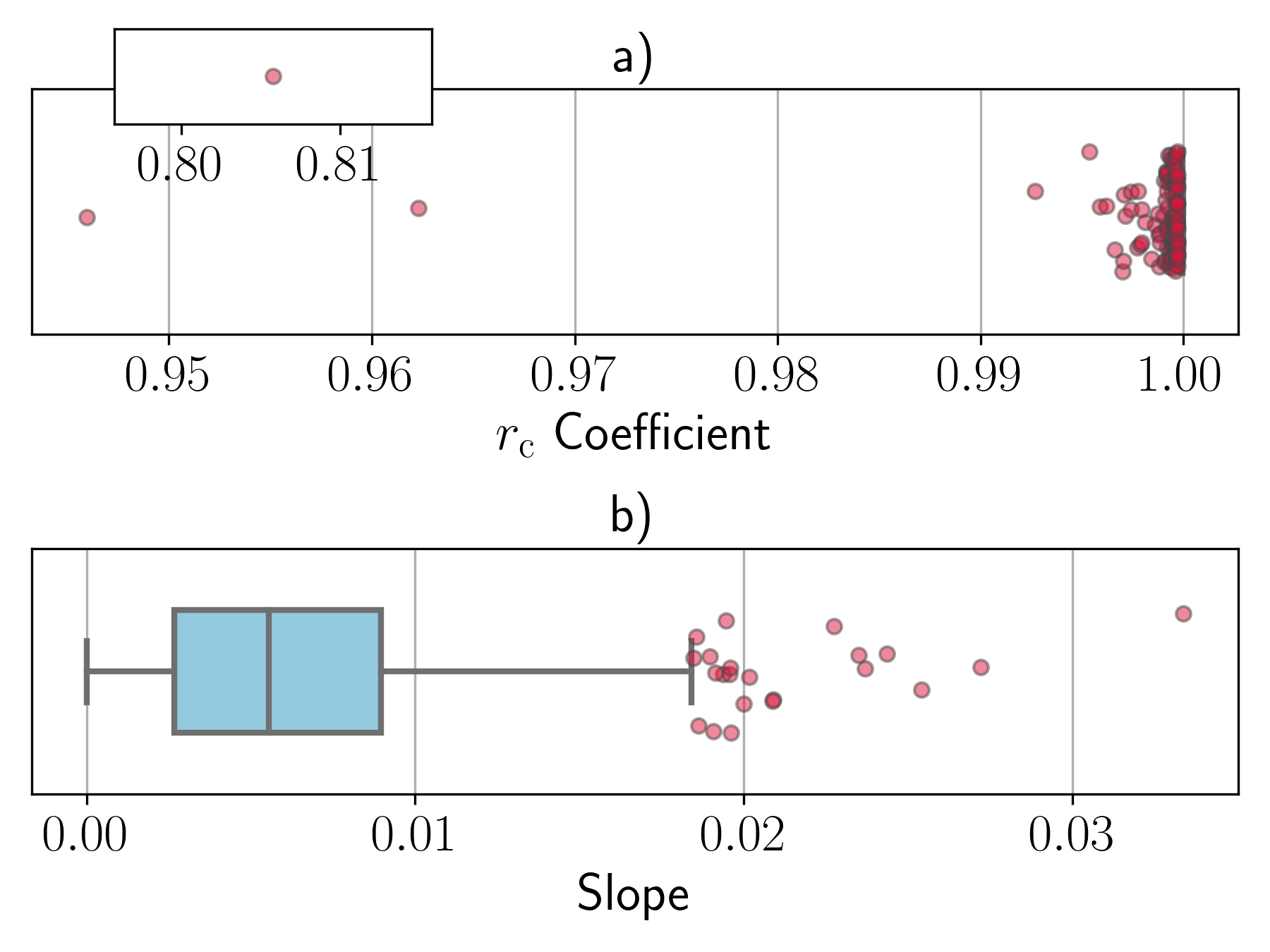}
    \caption{\label{fig:boxes_HE_Z}Statistics for the linear fits to the FI of 1050
    random models with $H_E=Z$. The FI of each data point is obtained with a prepare and
    measure strategy with noise from a random master equation. 
    (a) Slopes of the linear fits. The sub-plot in the top-left corner shows an isolated outlier and has been separated for aesthetic purposes. (b) Pearson coefficients $r_{\rm c}$ of the linear fits.  
    }
\end{figure}

\section{Conclusion}

In this work, we have considered the attainability of HL scaling in quantum metrology using QEC when the probe is subject to non-Markovian noise. We restricted our investigation to an HMM, which displays non-Markovianity for many choices of metric. We presented the HNES conditions, which guarantee the HL using QEC, and generalized the HNLS conditions of \cite{zhou_achieving_2018} to the non-Markovian regime. 
Contrary to the Markovian conditions, the HNES conditions are sufficient but not necessary. 
Through an example, we demonstrate that deriving a general necessary condition requires taking into account the initial state of the environment, and may potentially require tracking the system's trajectory. Thus, deriving such conditions is challenging.
We further present alternative conditions for obtaining the HL in certain scenarios where the HNES conditions may be violated. The trade-off is that achieving the HL in these scenarios may require a relatively strict measurement schedule. First, we show that when there are no dissipative terms in the master equation, i.e., when the evolution of the entire system is unitary, and when the signal acts only on the probe, then one can always achieve HL scaling of the QFI. Contrary to the HNES conditions, which guarantee exact quadratic scaling of the QFI, unitary evolution only guarantees quadratic scaling up to an almost periodic envelope function. In order to obtain the HL, one must therefore measure the system and specific recurrence times. Second, we showed that for a specific class of interactions, which we call diagonal interactions, the assumptions of the HNES conditions may be weakened. We presented the weaker HNELS conditions, which guarantee HL scaling of the QFI. Once again, the HL scaling is up to a periodic envelope function, and so measurements must be well-timed.

We show that the original HNLS conditions can be applied to the full system of probe and environment to guarantee that the SQL is the best achievable limit in certain cases. A summary of all the cases for the scaling of the QFI is presented in \cref{fig:venn_diagram}.

Finally, we investigate whether SQL scaling within a given experiment is always possible when we do not reset the state of the environment. We show that it is always achievable in a special case where the model has a diagonal interaction. 
We present numerical evidence of SQL scaling for a specific Heisenberg model where the operators do not commute and for a large number of randomly generated models. We consider a finite-size scaling argument on the FI, where the FI is calculated exactly for a few measurements and then via random sampling for a large number of measurements. When the signal acts on only the probe, there is strong evidence that the SQL is always achievable (see~\cref{fig:heisenberg}). When the signal also acts on the environment, the numerics are inconclusive due to the presence of important outliers in the data (see~\cref{fig:heisenberg_HE}).

Many open problems remain. First, there is the problem of deriving sufficient and necessary conditions for HL scaling. Second, a proof for the general attainability of the SQL in a given experiment is still lacking. 
Lastly, our work assumes a signal that acts unitarily and independently on the probe and environment. This setup is not completely general, and leaves open the problem of QFI scaling when the dissipative terms, as well as the unitary coupling, depend on the parameter to be estimated.

\section{Acknowledgements}

We would like to thank Senrui Chen, Matthew Duschenes, Noah Gorgichuk, Liang Jiang, Gideon Lee, Benjamin MacLellan, Akimasa Miyake, Arsalan Motamedi, Andy Schang, Ruhi Shah, Shravan Shravan, and Tyler Thurtell for helpful discussion. 
S.Z. and N.C. acknowledge funding
provided by Perimeter Institute for Theoretical Physics,
a research institute supported in part by the Government of Canada through the Department of Innovation, Science and Economic Development Canada, and by the
Province of Ontario through the Ministry of Colleges and
Universities.
S.Z. acknowledges the support from the National Research Council of Canada (Grant No.~AQC-217-1). 
N.C. acknowledges the support from the Applied Quantum Computing Challenge Program at the National Research Council of Canada.
R.L. would like to thank Mike and Ophelia Lazaridis for funding.

\section{Data Availability}

The data that support the findings of this article are openly available \cite{mann_zach-et-mnm-qec-metrology_2025}.

\bibliography{main}
    
\appendix

\section{KL Conditions for the HMM}\label{section:KL_conditions_proof}

We will follow the definitions and proofs of \cite{gottesman_surviving_2024,leung_university_2024}. 
We can modify Definition~\ref{def:def_qecc} to obtain an equivalent picture which includes the syndrome information, in which case we say $\lrp{U,\mathcal{E}}$ is a QECC if there exists an isometry $V$ such that
    \begin{equation*}
        (V\otimes\ident_{d'}
        )E(U\otimes \ident_d)\ket{\psi}\otimes\ket{\phi}=\ket{\psi}\otimes \sqrt{c_{E,\psi,\phi}}\ket{S_{E,\psi,\phi}},
    \end{equation*}
    where here the state $S$ contains both the syndrome and the environment degrees of freedom. This picture is obtained by taking the isometry $V$ to be the Stinespring dilation of the channel $D$. The states remain pure in this picture since the trace is not taken. Let $d_S$ be the dimension necessary for the dilation, such that $\ket{S_{E,\psi,\phi}}\in \C^{d_S}$.

    As for standard error correcting codes, if $(U,\mathcal{E})$ is a QECC, then $(U,\tx{span}(\mathcal{E}))$ is a QECC. The proof is the same as in the Markovian case. We follow the proof given by \cite{gottesman_surviving_2024,leung_university_2024}. 
    Consider $F,E\in\mathcal{E}$. As seen previously, there exists $V$ such that $\forall\ \ket{\psi}\in\C^k$ and $\forall\ \ket{\phi}\in\C^{d_E}$, we have that
    \begin{gather*}
        (V\otimes \ident_{d'_E})E(U\otimes \ident_{d_E})\ket{\psi}\ket{\phi}=\sqrt{c_{E,\psi,\phi}}\ket{\psi}\ket{S_{E,\psi,\phi}}, \\
        (V\otimes \ident_{d'_E})F(U\otimes \ident_{d_E})\ket{\psi}\ket{\phi}=\sqrt{c_{F,\psi,\phi}}\ket{\psi}\ket{S_{F,\psi,\phi}}.
    \end{gather*}
    By linearity, we have
    \begin{align*}
    (V\otimes \ident_{d'_E})(aE+bF)(U\otimes \ident_{d_E})\ket{\psi}\ket{\phi}\nonumber\\
    =\ket{\psi}\lrp{a\sqrt{c_{E,\psi,\phi}}\ket{S_{E,\psi,\phi}}+b\sqrt{c_{F,\psi,\phi}}\ket{S_{F,\psi,\phi}}}
    \end{align*}
    and so we see that $aE+bF$, where $a,b \in \mathbb{R}$ is also a correctable error, completing the proof. Using these properties, we can now prove Lemma~\ref{lemma:KL}.
    
    \begin{proof}[Proof~of~Lemma~\ref{lemma:KL}]
    The proof is similar to that of the standard KL conditions. We begin by showing that the conditions are necessary. Let $\ket{\psi},\ket{\phi}$ be linearly independent states of $\C^k$. Let $E_i,E_j$ be errors in $\mathcal{E}$. We consider orthogonal basis states $\ket{\chi_1},\ket{\chi_2}\in\C^{d_E}$. Since we have a QECC, then there exists an isometry $V$ such that 
    \begin{align}\label{eq:kl_proof_1}
        (V\otimes \ident_{d'_E})E_i(U\otimes \ident_{d_E})\ket{\psi}\ket{\chi_1}\nonumber\\
        =\sqrt{c_{E_i,\psi,\chi_1}}\ket{\psi}\ket{S_{E_i,\psi,\chi_1}},
    \end{align}
    \begin{align}\label{eq:kl_proof_2}
        (V\otimes \ident_{d'})E_j(U\otimes \ident_d)\ket{\phi}\ket{\chi_2}\nonumber\\
        =\sqrt{c_{E_j,\phi,\chi_2}}\ket{\phi}\ket{S_{E_j,\psi,\chi_2}}.
    \end{align}
    We now consider a state of the form 
    \begin{equation*}
        \ket{\eta}=\alpha\ket{\psi}+\beta\ket{\psi^{\perp}}
    \end{equation*}
    where $\ket{\psi^{\perp}}$ is a vector that is orthogonal to $\ket{\psi}$. We have that 
    \begin{align}\label{eq:kl_proof_3}
        (V\otimes \ident_{d'_E})E_i\lrp{U\otimes\ident_{d_E}}\ket{\eta}\ket{\chi_1} \nonumber\\
        =\sqrt{c_{E_i,\eta,\chi_1}}\ket{\eta}\ket{S_{E_i,\eta,\chi_1}}.
    \end{align}
    We now project with $\ket{\psi}\bra{\psi}$ on both \cref{eq:kl_proof_1} and \cref{eq:kl_proof_3} and then divide the second equation by $\alpha$. The left-hand sides become equal, and we equate them to obtain
    \begin{equation*}
        \sqrt{c_{E_i,\psi,\chi_1}}\ket{S_{E_i,\psi,\chi_1}}=\sqrt{c_{E_i,\eta,\chi_1}}\ket{S_{E_i,\eta,\chi_1}}\quad \forall \ \eta.
    \end{equation*}
    We can therefore conclude that the residual syndrome and environment state do not depend on the choice of initial state, i.e.
    \begin{equation*}
        c_{E_i,\psi,\chi_1}\equiv c_{E_i,\chi_1}.
    \end{equation*}
    Knowing this, we now take the inner product between \cref{eq:kl_proof_1} and \cref{eq:kl_proof_3} with the projector $\ket{\xi_l}_S\bra{\xi_k}_S$ sandwiched between the bra and ket, where $\ket{\xi_l}_S,\ket{\xi_k}_S$ are orthogonal basis elements of $\C^{d_S}$. We also relabel $\chi_1$ to $\chi_m$ and $\chi_2$ to $\chi_n$. This gives us the new equation 
    \begin{align*}
        &\bra{\phi}\bra{\chi_m}(U^{\dagger}\otimes \ident_{d_E})E_j^{\dagger}(V^{\dagger}\otimes \ident_{d'_E})\nonumber\\ &\qquad \cdot(\ident_k\otimes\ket{\xi_l}_S\bra{\xi_k}_S) (V\otimes\ident_{d'_E}) E_i (U\otimes \ident_{d_E})\ket{\psi}\ket{\chi_n}\nonumber\\
        &=\sqrt{c_{E_i,\chi_n}c_{E_j,\chi_m}}\bra{S_{E_j,\chi_m}}\ket{\xi_l}\bra{\xi_k}\ket{S_{E_i,\chi_n}}\braket{\phi}{\psi}.
    \end{align*}
    By defining 
    \begin{align*}
    c_{ij}^{mnlk}=\sqrt{c_{E_i,\chi_m}c_{E_j,\chi_n}}\bra{S_{E_j,\chi_n}}\ket{\xi_l}\nonumber\\
    \times\bra{\xi_k}\ket{S_{E_i,\chi_m}}\braket{\phi}{\psi},
    \end{align*}
    the equation is reduced to 
    \begin{equation*}
        P\bra{\chi_m}E_i^{\dagger}\ket{\chi_l}\bra{\chi_k}E_j\ket{\chi_n}P=c_{ij}^{mnlk}P.
    \end{equation*}
    This is exactly \cref{eq:KLM_HMM} for the KL conditions, thus completing the proof that the conditions are necessary for a QECC. 

    We now show that the conditions are sufficient. By defining
    \begin{equation}\label{eq:extended_errors}
    \Tilde{E}_{i,n,k}=\bra{n}_{E'}E_i\ket{k}_{E},
    \end{equation}
    one can see that the generalized conditions are equivalent to the standard KL conditions of \cite{knill_theory_1996}, where we consider the extended error set $\lrcb{\tilde{E}_{i,n,k}}$ indexed by $i,n,k$ acting solely on the probe.

    Let $\ket{\chi}\in\C^{d_E}$ be the environment state and $\ket{\psi}\in\C^{d_P}$ be a state in the code space. We expand the environment state over an orthonormal basis to obtain
    \begin{equation*}
        \ket{\chi}=\sum_{l=0}^{d_E}c_l\ket{l}_{E}.
    \end{equation*}
    Let $E_i\in\mathcal{E}$ be an error. We then have that
    \begin{align*}
        &\quad E_i\ket{\psi}\ket{\chi} \nonumber\\
        &=\!\bigg(\!\sum_{n=0}^{d_E'}\sum_{k=0}^{d_E} \bra{n}_{E'}E_i\ket{k}_E \!\otimes\! \ket{n}_{E'}\bra{k}_E\!\bigg)\lrp{\!\ket{\psi} \!\otimes\! \sum_{l=0}^{d_E}c_l\ket{l}_E\!} \nonumber\\
        &=\sum_{n=0}^{d_E'}\lrb{\lrp{\sum_{k=0}^{d_E} c_k\Tilde{E}_{i,n,k}\ket{\psi}}\otimes \ket{n}_{E'}}.
    \end{align*}
    Since we suppose the standard KL conditions hold for $\Tilde{E}_{i,n,k}$ on the probe, then by \cite{knill_theory_1996} there exists a CPTP map $D$ that acting solely on the probe such that 
    \begin{align*}
        (D\otimes \ident_{d'_E})\left(\sum_{n,n'=0}^{d_E'}\left[\left(\sum_{k,k'=0}^{d_E} c_kc_{k'}^*\Tilde{E}_{i,n,k}\ket{\psi}\bra{\psi}\tilde{E}_{i,n',k'}^{\dagger}\right)\right.\right.\nonumber\\ \otimes \ket{n}_{E'}\bra{n'}_{E'}\bigg]\bigg)
        =\ket{\psi}\bra{\psi}\otimes \sum_{n,n'=0}^{d_E'}\alpha_{nn'}\ket{n}_{E'}\bra{n'}_{E'}\quad\quad
    \end{align*}
    where $\alpha_{nn'}$ are the coefficients of the residual non-normalized state in the environment. This holds since the $D$ map can correct for linear combinations of the errors $\Tilde{E}_{i,n,k}$. We need only compose with a partial trace map on the environment degrees of freedom to obtain the exact definition for a QECC. The measurement and recovery operators for the HMM are obtained from the code for the reduced errors $\Tilde{E}_{i,n,k}$. 
    \end{proof}

\section{Example Noise in the Markov Approximation}\label{section:example_markov_approx}

 We will consider $N$ environmental spin-1 degrees of freedom. We will consider a single spin-half, as we will assume the noise on each spin-half is uncorrelated.
The total Hamiltonian of the system can be written as 
\begin{align*}
H=\sum_{i=1}^NH_{E,i}+\sum_{i=1}^N H_{EP,i}
\end{align*}
where
\begin{equation*}
    H_{EP,i}=\tau_i\lrp{ S_x^{(i)}\otimes X+S_z^{(i)}\otimes Z}.
\end{equation*}
Here, the superscript is used to index the spins the operator acts on. We will suppose without loss of generality that the total Hamiltonian has no non-trivial terms acting on the probe. If this were not the case, then these terms could be eliminated by the fast and accurate universal control. This spin-1 Hamiltonian is inspired by the Hamiltonian of nitrogen vacancy centers in diamond \cite{liddy_optimal_2022,rama_koteswara_rao_level_2020}. The full interaction Hamiltonian can be written 
\begin{align*}
H_{I}=\sum_{i=1}^N\sum_{A\in\lrcb{X,Z}} \tau_{i} S_A^{(i)}\otimes A.
\end{align*}
We set 
\begin{equation}\label{eq:zero_split_Hamiltonian}
    H_{E,i}=\Delta (S_z^{(i)})^2,
\end{equation}
where $\Delta$ is the zero-field splitting. $\tau_{i}$ is the strength of the coupling between environmental spin $i$ and the spin-half, which we assume to be small. In the Markovian regime, we will assume $\Delta\gg \tau_i$, such that the effect of the coupling on the environment state can be neglected. We see that $\ket{0}_E$ is the ground state of the zero-field splitting Hamiltonian given by~\cref{eq:zero_split_Hamiltonian}. Following the prescription of \cite{breuer_measure_2009}, we can obtain a Markovian master equation acting on the spin-half by making the Born-Markov approximation. In this approximation, we assume that the state can be written as
\begin{equation}\label{eq:theory_OQS_master_eq}
    \rho(t)\approx \rho_P(t)\otimes \lrp{\ket{0}_E\bra{0}_E}^{\otimes N}.
\end{equation}
In other words, the environment spins are approximated to always be in the state $\lrp{\ket{0}_E\bra{0}_E}^{\otimes N}$ due to their weak coupling and fast decoherence.
The master equation acting on the probe is then given by Eq. (3.143) of \cite{breuer_measure_2009}.
\begin{align}
    \frac{d\rho_P}{dt}=\sum_{A,Q\in\lrcb{X,Z}}\gamma_{AQ}\lrp{A\rho_P Q-\frac{1}{2}\lrcb{QA,\rho_P}}.
\end{align}
We have neglected the Lamb shift Hamiltonian as it can be offset by our controls, leaving only the dissipative terms. The decay rates $\gamma_{AQ}$ are given by Eq. (3.139) of \cite{breuer_theory_2007}
\begin{align*}
    &\gamma_{AQ}=\\
    &\sum_{i,j=1}^N\int_{\R}dt e^{\ii\omega t}\Tr\lrb{e^{\ii t H_{E,i}}S_A^{(i)}e^{-\ii t H_{E,i}}S_Q^{(j)} \lrp{\ket{0}_E\bra{0}_E}^{\otimes N}}.
\end{align*}
In particular, we see that $\gamma_{AQ}=0$ if $A$ or $Q$ is equal to $Z$. And so,~\cref{eq:theory_OQS_master_eq} reduces to 
\begin{equation*}
    \frac{d\rho_P}{dt}=\gamma_{XX}\lrp{X\rho_PX - \rho}.
\end{equation*}
We recognize the equation for dephasing along the $X$ axis. The corresponding channel has Kraus operators proportional to $\ident$ and $X$.

\section{Direct Proof of the Limit Superior for the Unitary Case}\label{section:direct_alpha_proof}

It is sufficient to show for any $t_0 > 0$ and $\epsilon > 0$, there exists $t > t_0$ such that $\abs{\alpha(t)-\alpha(0)} = \abs{\alpha(t)-\frac{1}{2}} < \epsilon$. First, due to the continuity and periodicity of $e^{i\phi}$ as a function of $\phi$, there is a constant $\delta$ that depends on $\epsilon$ such that $\abs{\sum_m\bar{c}_me^{-\ii t\phi_m} - \sum_m\bar{c}_m} < \epsilon$ whenever $\abs{t\phi_m - 2k_m\pi}\leq \delta$ for some set of integers $\{k_m\}$ and all $m$. Without loss of generality, we also assume $\delta$ is small enough such that $2\pi \lceil 1/(4\pi\delta) \rceil > t_0$ and $4\pi\delta^2 < \delta$.  Second, according to Dirichlet's approximation theorem~\cite{Dirichlet}, there exists an integer $1 \leq q \leq \lceil 1/\delta^2 \rceil^{d_E^2}$ and a set of integers $\{q_m\}$ such that $\abs{\phi_m - q_m/q} \leq \frac{1}{\lceil 1/\delta^2 \rceil q} \leq \delta^2/q$. Finally, let $t = 2\pi \lceil 1/(4\pi\delta) \rceil q $. We have 
\begin{equation}
    \abs{t\phi_m - t\frac{q_m}{q}} \leq  \left\lceil \frac{1}{4\pi\delta} \right\rceil q 2\pi \frac{\delta^2}{q} \leq \frac{\delta + 4\pi\delta^2}{2} < \delta. 
\end{equation}
The above provides a construction of $t$ that satisfies $t = 2\pi \lceil 1/(4\pi\delta) \rceil q > t_0$ and $\abs{t\phi_m - t\frac{q_m}{q}} < \delta$ as required, which implies $\abs{\sum_m\bar{c}_me^{-\ii t\phi_m} - \sum_m\bar{c}_m} < \epsilon$. Therefore, we have $\limsup_{t\rightarrow\infty}\abs{\alpha(t)}=\frac{1}{2}$, leading to a QFI scaling of  
\begin{equation}
\limsup_{t\rightarrow \infty}\frac{\mathcal{F}(\rho_{PA})}{t^2}=\lrp{\Delta\lambda}^2. 
\end{equation}

\section{Proof of linear FI scaling for the binomial mixture}\label{section:binomial_linear_FI}

The FI, which we denote by $\mathcal{I}$, can be calculated using the following formula 
\begin{align}
    \mathcal{I}(\omega) &=\sum_{b_1,\cdots,b_N}\lrp{\frac{\del_{\omega}P(b_1,\cdots,b_N)}{P(b_1,\cdots,b_N)}}^2 P(b_1,\cdots,b_N) \nonumber\\
    & =\mathbb{E}\lrb{\lrp{\frac{\del_{\omega}P(b_1,\cdots,b_N)}{P(b_1,\cdots,b_N)}}^2}. \label{eq:classical_FI}
\end{align}

We first consider the numerator of \cref{eq:classical_FI} and take the derivative of \cref{eq:binomial_mixture},
\begin{align}\label{eq:diff_pi}
    &\del_{\omega}P(N_1=n_1)=\nonumber \\
    &\sum_{i=1}^{d_E}\alpha_i\binom{N}{n_1}p_i^{n_1}(1-p_i)^{N-n_1}\frac{n_1-Np_i}{p_i(1-p_i)}\lrp{\del_{\omega}p_i}.
\end{align}
We will choose $t>0$ such that $p_i\neq 0$ or 1. We can see from \cref{eq:pi_analytic} that for $t=0$, $p_i=0.5$. Further, the probability depends continuously on the variable $t$. There therefore exists a choice $t>0$ such $p_i$ is not too far from $\frac{1}{2}$, for all $i$. Similarly, we will choose $t$ such that $\del_{\omega} p_i\neq 0$. We can see from \cref{eq:diff_pi_analytic} that for $t>0$ but small, the derivative is given by 
\begin{align*}
    \del_{\omega}p_i=\frac{1}{2}\exp\lrp{-\Re(\Gamma_i)t}(\lambda_0-\lambda_1)t+O(t^2),
\end{align*}
and so for a choice of small $t$ the derivative doesn't vanish. Below, we assume $t$ is small enough so that for all $i$, $ 0 < p_i < 1$ and $\del_{\omega}p_i \neq 0$.

The full FI can be written as
\begin{align}\label{eq:FI_mixture}
    &\mathcal{I}(\omega)= \nonumber \\
    &\sum_{n_1=0}^{N}\frac{\lrp{\sum_{i=1}^{d_E}\alpha_i(\del_{\omega}p_i)p_i^{n_1}(1-p_i)^{N-n_1}\frac{n_1-Np_i}{p_i(1-p_i)}}^2}{\lrp{\sum_{i=1}^{d_E}\alpha_i p_i^{n_1}(1-p_i)^{N-n_1}}^2}\nonumber \\
    &\cdot \lrp{\binom{N}{n_1}\sum_{i=1}^{d_E}\alpha_i p_i^{n_1}(1-p_i)^{N-n_1}}.
\end{align}

To show the asymptotic linear scaling of the FI, we will need to prove three important bounds. Let $i_*=\tx{argmax}\lrcb{p_i\ | \ \alpha_i> 0}$ and let $p_*=p_{i_*}$. Such a $p_i$ always exists since we've assumed $\sum_i\alpha_i=1$, and so there must be at least one nonzero $\alpha_i$.  
To prove the SQL is always achievable, we will need the following three bounds. The proofs of the bounds can be found below.

\begin{lemma}\label{lemma:first_bound}
$\forall \ \eta_1>0,\epsilon_1>0$, and $\forall\ p_i\neq p_*$, $\exists \ $ $N_1^{(1)}$ such that $\forall \ N>N_1^{(1)}$, and $\forall \ n_1$ satisfying
\begin{equation*}
    \frac{n_1}{N}>\frac{\log\lrp{\frac{1-p_i}{1-p_*}}}{\log\lrp{\frac{p_*(1-p_i)}{p_i(1-p_*)}}}+\epsilon_1,
\end{equation*}
we have that 
\begin{align}
    \abs{\alpha_*p_*^{n_1}(1-p_*)^{N-n_1}
    \cdot\frac{n_1-Np_*}{p_*(1-p_*)}(\del_{\omega}p_*)}\nonumber\\ >\eta_1\abs{ \alpha_ip_i^{n_1}(1-p_i)^{N-n_1}
    \cdot\frac{n_1-Np_i}{p_i(1-p_i)}(\del_{\omega}p_i)}. \label{eq:first_bound}
\end{align}
\end{lemma}

Note that the right-hand side of this equation satisfies
\begin{equation*}
    \frac{\log\lrp{\frac{1-p_i}{1-p_*}}}{\log\lrp{\frac{p_*(1-p_i)}{p_i(1-p_*)}}} + \epsilon_1 > p_i.
\end{equation*}

\begin{lemma}\label{lemma:second_bound}
$\forall\ \eta_2>0,\ \epsilon_2>0$, and $\forall\ p_i\neq p_*$, $\exists \ N_1^{(2)}$ such that $\forall \ N>N_1^{(2)}$, and $\forall\ n_1$ satisfying
\begin{equation*}
    \frac{n_1}{N}>\frac{\log\lrp{\frac{1-p_i}{1-p_*}}}{\log\lrp{\frac{p_*(1-p_i)}{p_i(1-p_*)}}}+\epsilon_2,
\end{equation*}
we have that
\begin{align}
    \alpha_* p_*^{n_1}(1-p_*)^{N-n_1}>\eta_2 \alpha_i p_i^{n_1}(1-p_i)^{N-n_1}.\label{eq:second_bound}
\end{align}
\end{lemma}
\begin{lemma}\label{lemma:third_bound}
$\forall\ 1>\eta_3>0,\ \epsilon_3>0$, $\exists\ N_1^{(3)}$ such that $\forall \ N>N_1^{(3)}$, $\exists\ n_1^{(3)}>(p_*-\epsilon_3)N$ such that $n_1>n_1^{(3)}$ implies that
\begin{align}
    \sum_{n_1>n_1^{(3)}}\binom{N}{n_1}\frac{\lrp{p_*^{n_1}(1-p_*)^{N-n_1}\frac{n_1-Np_*}{p_*(1-p_*)}}^2}{p_*^{n_1}(1-p_*)^{N-n_1}} \nonumber \\
    > \eta_3\sum_{n_1=1}^{N}\binom{N}{n_1}\frac{\lrp{p_*^{n_1}(1-p_*)^{N-n_1}\frac{n_1-Np_*}{p_*(1-p_*)}}^2}{p_*^{n_1}(1-p_*)^{N-n_1}}. \label{eq:third_bound}
\end{align}
\end{lemma}
We can combine these three bounds to show linear scaling. Let $p_{2}=\max\lrcb{p_i\ | \ \alpha_i\neq 0, \, p_i<p_*}$ be the second largest $p_i$. If such a $p_i$ doesn't exist, then there is only one binomial distribution in the mixture. For that particular case, we know that the FI of a binomial distribution scales linearly in $N$, completing the proof that the SQL is achievable. Indeed the FI of a binomial distribution with parameter $p_{*}$ is given by
\begin{align}\label{eq:binomial_FI}
    \mathcal{I}(\omega)=\frac{N}{p_*(1-p_*)}\lrp{\del_{\omega}p_*}^2.
\end{align}

We now consider the case where $p_2$ is well defined. We choose $\epsilon_1,\epsilon_2,\epsilon_3$ (as in Lemmas~\ref{lemma:first_bound},\ref{lemma:second_bound}, and~\ref{lemma:third_bound} respectively) small enough such that
\begin{align*}
    \frac{\log\lrp{\frac{1-p_2}{1-p_*}}}{\log\lrp{\frac{p_*(1-p_2)}{p_2(1-p_*)}}}+\max\lrcb{\epsilon_1,\epsilon_2}<p_*-\epsilon_3. 
\end{align*}
We show this inequality holds in Lemma~\ref{lemma:loglog_pstar_lemma}. Then, for any $\eta_1,\eta_2>d_E$ and $1>\eta_3>0$, we know that for any $N$ large enough there exists a corresponding $n_1^{*}$ such that
\begin{align*}
    p_2<\frac{\log\lrp{\frac{1-p_2}{1-p_*}}}{\log\lrp{\frac{p_*(1-p_2)}{p_2(1-p_*)}}}+\max\lrcb{\epsilon_1,\epsilon_2} < \frac{n_1^*}{N}<p_*-\epsilon_3,
\end{align*}
and such that \cref{eq:first_bound}, \cref{eq:second_bound}, and \cref{eq:third_bound} hold.
Applying these bounds allows us to derive the following inequalities for the FI. 

We first remove a few terms of the sum of \cref{eq:FI_mixture}. Since all the terms are positive, this leads to a lower bound.
\begin{align*}
    &\mathcal{I}(\omega)\geq \nonumber \\
    &\sum_{n_1>n_1^*}^{N}\binom{N}{n_1}\frac{\lrp{\sum_{i=1}^{d_E}\alpha_i(\del_{\omega}p_i)p_i^{n_1}(1-p_i)^{N-n_1}\frac{n_1-Np_i}{p_i(1-p_i)}}^2}{\lrp{\sum_{i=1}^{d_E}\alpha_i p_i^{n_1}(1-p_i)^{N-n_1}}}.
\end{align*}
Next, we apply Lemma~\ref{lemma:first_bound} to lower bound the numerator.
\begin{align*}
    &\mathcal{I}(\omega)\geq 
    \sum_{n_1>n_1^{*}}\binom{N}{n_1}\cdot \\
    &\frac{\lrp{(1-d_E\eta_1^{-1})\alpha_*(\del_{\omega}p_*)p_*^{n_1}(1-p_*)^{N-n_1}\frac{n_1-Np_*}{p_*(1-p_*)}}^2}{\sum_{i=1}^{d_E}\alpha_ip_i^{n_1}(1-p_i)^{N-n_1}},
\end{align*}
where we have upper bounded each of the terms associated with $p_i$ with $p_i\neq p_*$ by $\eta_1^{-1}$ times the $p_*$ term. We give each of these terms a sign opposite to that of the $p_*$ term, so that we obtain a lower bound. We now apply Lemma~\ref{lemma:second_bound} to bound the denominator.
\begin{align*}
    &\mathcal{I}(\omega)\geq \sum_{n_1>n_1^{*}}\binom{N}{n_1}\cdot \\
    &\frac{\lrp{(1-d_E\eta_1^{-1})\alpha_*(\del_{\omega}p_*)p_*^{n_1}(1-p_*)^{N-n_1}\frac{n_1-Np_*}{p_*(1-p_*)}}^2}{(1+d_E\eta_2^{-1})\alpha_*p_*^{n_1}(1-p_*)^{N-n_1}}.
\end{align*}
We similarly bound the denominator by bounding all the non-$p_*$ terms by $\eta_2^{-1}$ times the $p_*$ term, and assigning them the same sign to obtain a lower bound. We can rewrite this bound as
\begin{align*}
    &\mathcal{I}(\omega)\geq\frac{(1-d_E\eta_1^{-1})^2}{1+d_E\eta_2^{-1}}\cdot \\
    &\sum_{n_1>n_1^{*}}\binom{N}{n_1}\frac{\lrp{\alpha_*(\del_{\omega}p_*)p_*^{n_1}(1-p_*)^{N-n_1}\frac{n_1-Np_*}{p_*(1-p_*)}}^2}{\alpha_*p_*^{n_1}(1-p_*)^{N-n_1}}.
\end{align*}
Finally, we apply Lemma~\ref{lemma:third_bound} to the inequality to obtain
\begin{align*}
    &\mathcal{I}(\omega)\geq\frac{(1-d_E\eta_1^{-1})^2}{1+d_E\eta_2^{-1}}\eta_3 \cdot \\
    &\sum_{n_1=0}^{N}\binom{N}{n_1}\frac{\lrp{\alpha_*(\del_{\omega}p_*)p_*^{n_1}(1-p_*)^{N-n_1}\frac{n_1-Np_*}{p_*(1-p_*)}}^2}{\alpha_*p_*^{n_1}(1-p_*)^{N-n_1}}.
\end{align*}
We recognize the formula for the FI of a binomial distribution, and can therefore substitute the known result, \cref{eq:binomial_FI}, to obtain a final lower bound for the FI of the binomial mixture.
\begin{align*}
    \mathcal{I}(\omega) \geq \eta_3\frac{(1-d_E\eta_1)^2}{1+d_E\eta_2}\alpha_*^2(\del_{\omega}p_*)^2 \frac{N}{p_*(1-p_*)}.
\end{align*}

We see that the final lower bound scales linearly in $N$, thus completing the proof that the SQL is achievable.

\subsection{Proof of Lemma~\ref{lemma:first_bound}}

We will show that $\forall \ \eta_1>0,\epsilon_1>0$, and $\forall\ p_i\neq p_*$, $\exists \ $ $N_1^{(1)}$ such that $\forall \ N>N_1^{(1)}$, and $\forall \ n_1$ satisfying
\begin{equation}
    \frac{n_1}{N}>\frac{\log\lrp{\frac{1-p_i}{1-p_*}}}{\log\lrp{\frac{p_*(1-p_i)}{p_i(1-p_*)}}}+\epsilon_1,\label{eq:first_bound_appendix_n1}
\end{equation}
we have that 
\begin{align}
    \abs{\alpha_*p_*^{n_1}(1-p_*)^{N-n_1}
    \cdot\frac{n_1-Np_*}{p_*(1-p_*)}(\del_{\omega}p_*)}\nonumber \\ >\eta_1\abs{ \alpha_ip_i^{n_1}(1-p_i)^{N-n_1}
    \cdot\frac{n_1-Np_i}{p_i(1-p_i)}(\del_{\omega}p_i)}.\label{eq:first_bound_appendix}
\end{align}

\begin{proof}
    
We can rearrange the terms of \cref{eq:first_bound_appendix} to obtain the following equation
\begin{align*}
    \abs{\frac{p_*(1-p_i)}{p_i(1-p_*)}}^{n_1}\abs{\frac{1-p_*}{1-p_i}}^N\abs{\frac{n_1-Np_*}{n_1-Np_i}} \\
    > \eta_1\frac{\alpha_i}{\alpha_*}\frac{p_*(1-p_*)}{p_i(1-p_i)}\abs{\frac{\del_{\omega}p_i}{\del_{\omega}p_*}}.
\end{align*}
We will denote the right side of this equation by $\hat{\eta}$. Remark $\hat{\eta}$ depends on neither $N$ or $n_1$. Taking the natural log of both sides, we obtain
\begin{align*}
    n_1\log\abs{\frac{p_*(1-p_i)}{p_i(1-p_*)}} + N\log\abs{\frac{1-p_*}{1-p_i}}\\ 
    +\log\abs{\frac{n_1-Np_*}{n_1-Np_i}}>\log\hat{\eta}.
\end{align*}
First, we let $n_1=Nk_1$. We divide the previous equation by $N$ and shuffle terms around to obtain 
\begin{align*}
    &k_1>\lrp{\log\abs{\frac{p_*(1-p_i)}{p_i(1-p_*)}}}^{-1}\\
    &\cdot \lrp{\log\abs{\frac{1-p_i}{1-p_*}}+\frac{1}{N}\log\abs{\frac{k_1-p_i}{k_1-p_*}}+\frac{\log\hat{\eta}}{N}}.
\end{align*}
There exists an $N_1^{(1)}$ such that $N>N_1^{(1)}$ implies that the last two terms when expanding the right side of the equation are smaller than $\epsilon_1$. This yields the inequality
\begin{align*}
    k_1>\frac{\log\lrp{\frac{1-p_i}{1-p_*}}}{\log\lrp{\frac{p_*(1-p_i)}{p_i(1-p_*)}}} + \epsilon_1.
\end{align*}
And so, the condition for \cref{eq:first_bound_appendix} to hold is that
\begin{align*}
    k_1=\frac{n_1}{N}>\frac{n_1^{(1)}}{N},\quad N>N_1^{(1)},
\end{align*}
where we let
\begin{align}
    \frac{n_1^{(1)}}{N}=\frac{\log\lrp{\frac{1-p_i}{1-p_*}}}{\log\lrp{\frac{p_*(1-p_i)}{p_i(1-p_*)}}} + \epsilon_1.
\end{align}

\end{proof}

In the proofs of the main text, we use the following property of $n_1^{(1)}$.
\begin{lemma}\label{lemma:loglog_pstar_lemma}
For $p_i<p_*$ and $0<p_i,p_*<1$, we have
\begin{equation*}
    \frac{\log\lrp{\frac{1-p_i}{1-p_*}}}{\log\lrp{\frac{p_*(1-p_i)}{p_i(1-p_*)}}}<p_*.
\end{equation*}
\end{lemma}
\begin{proof}
    
We can rewrite this equation as 
\begin{equation*}
    \log\lrp{\frac{1-p_i}{1-p_*}}-p_*\log\lrp{\frac{p_*(1-p_i)}{p_i(1-p_*)}}<0.
\end{equation*}
Taking the exponential of both sides, we obtain the inequality
\begin{equation}\label{eq:pstar_log_proof_eq}
    \exp\lrp{\log\lrp{\frac{1-p_i}{1-p_*}}-p_*\log\lrp{\frac{p_*(1-p_i)}{p_i(1-p_*)}}}<1.
\end{equation}
We will now prove this inequality. First, remark that for $p_i=p_*$, we have that
\begin{equation*}
\exp\lrp{\log\lrp{\frac{1-p_*}{1-p_*}}-p_*\log\lrp{\frac{p_*(1-p_*)}{p_*(1-p_*)}}}=1.
\end{equation*}
We now consider the derivative of the left-hand side with respect to $p_*$.
\begin{align*}
    \frac{\del}{\del p_*}\exp\lrp{\log\lrp{\frac{1-p_i}{1-p_*}}-p_*\log\lrp{\frac{p_*(1-p_i)}{p_i(1-p_*)}}}= \\
    -\lrp{\frac{1-p_i}{1-p_*}}^{1-p_*}\lrp{\frac{p_i}{p_*}}^{p_*}\lrb{\log\lrp{\frac{p_*}{p_i}}+\log\lrp{\frac{1-p_i}{1-p_*}}}.
\end{align*}
We see that the derivative is strictly negative for all $p_i<p_*$ such that $0<p_i,p_*<1$. In other words, for $p_*>p_i$, as we increase $p_*$, the function defined by the left side of \cref{eq:pstar_log_proof_eq} decreases from the value of $1$ at $p_i=p_*$. Therefore, the inequality of \cref{eq:pstar_log_proof_eq} must hold.

\end{proof}

Equivalently, by the density of real numbers, this lemma states that there exists an $\epsilon>0$ such that 
\begin{equation*}
    \frac{\log\lrp{\frac{1-p_i}{1-p_*}}}{\log\lrp{\frac{p_*(1-p_i)}{p_i(1-p_*)}}}<p_*-\epsilon.
\end{equation*}

\subsection{Proof of Lemma~\ref{lemma:second_bound}}
We now consider the second bound. We show that $\forall\ \eta_2>0$, $\forall\ p_i\neq p_*$, $\exists \ N_1^{(2)}$ such that $\forall \ N>N_1^{(2)}$, $\exists \ n_1^{(2)}$ such that $n_1>n_1^{(2)}$ implies
\begin{align*}
    \alpha_* p_*^{n_1}(1-p_*)^{N-n_1}>\eta_2\alpha_i p_i^{n_1}(1-p_i)^{N-n_1}.
\end{align*}
\begin{proof}
    
We can rearrange the equation by following the same steps as the previous derivation to obtain 
\begin{align*}
    k_1>\frac{\log\lrp{\frac{1-p_i}{1-p_*}}}{\log\lrp{\frac{p_*(1-p_i)}{p_i(1-p_*)}}} + \frac{\log\eta_2}{N}.
\end{align*}
For all $\epsilon_2>0$, there exists $N_1^{(2)}$ such that for $N>N_1^{(2)}$, the last term is smaller than $\epsilon_2$. This gives us 
\begin{align}
    k_1>\frac{\log\lrp{\frac{1-p_i}{1-p_*}}}{\log\lrp{\frac{p_*(1-p_i)}{p_i(1-p_*)}}} + \epsilon_2.
\end{align}
Once again, the choice of $N_1^{(2)}$ is somewhat arbitrary, as we can choose to suppress $\epsilon_2$ more or less in exchange for lowering the value of $n_1^{(2)}$, which we define as 
\begin{align*}
    \frac{n_1^{(2)}}{N}=\frac{\log\lrp{\frac{1-p_i}{1-p_*}}}{\log\lrp{\frac{p_*(1-p_i)}{p_i(1-p_*)}}} + \epsilon_2.
\end{align*}
Once again, we must choose $N_1^{(2)}$ big enough such that $n_1^{(2)}<N$ and such that there exists integer values of $n_1$ such that $n_1^{(2)}<n_1<N$.

\end{proof}

\subsection{Proof of Lemma~\ref{lemma:third_bound}}

Finally, we prove a final bound. We will show that $\forall\ 1>\eta_3>0,\ \epsilon_3>0$, $\exists\ N_1^{(3)}$ such that $\forall \ N>N_1^{(3)}$, $\exists\ n_1^{(3)}>(p_*-\epsilon_3)N$ such that $n_1>n_1^{(3)}$ implies that
\begin{align}
    \sum_{n_1>n_1^{(3)}}\binom{N}{n_1}\frac{\lrp{p_*^{n_1}(1-p_*)^{N-n_1}\frac{n_1-Np_*}{p_*(1-p_*)}}^2}{p_*^{n_1}(1-p_*)^{N-n_1}} \nonumber \\
    > \eta_3\sum_{n_1=1}^{N}\binom{N}{n_1}\frac{\lrp{p_*^{n_1}(1-p_*)^{N-n_1}\frac{n_1-Np_*}{p_*(1-p_*)}}^2}{p_*^{n_1}(1-p_*)^{N-n_1}}.
\end{align}
\begin{proof}
    
We can decompose the RHS into partial sums,
\begin{align*}
    \tx{RHS}=\eta_3\sum_{n_1\leq n_1^{(3)}}\binom{N}{n_1} \frac{\lrp{p_*^{n_1}(1-p_*)^{N-n_1}\frac{n_1-Np_*}{p_*(1-p_*)}}^2}{p_*^{n_1}(1-p_*)^{N-n_1}} \\
    + \eta_3\sum_{n_1>n_1^{(3)}}\binom{N}{n_1}\frac{\lrp{p_*^{n_1}(1-p_*)^{N-n_1}\frac{n_1-Np_*}{p_*(1-p_*)}}^2}{p_*^{n_1}(1-p_*)^{N-n_1}}.
\end{align*}
We can bound the first term,
\begin{align*}
    \eta_3\sum_{n_1\leq n_1^{(3)}} \binom{N}{n_1}\frac{\lrp{p_*^{n_1}(1-p_*)^{N-n_1}\frac{n_1-Np_*}{p_*(1-p_*)}}^2}{p_*^{n_1}(1-p_*)^{N-n_1}} \\
    =\eta_3\sum_{n_1\leq n_1^{(3)}} \binom{N}{n_1}p_*^{n_1}(1-p_*)^{N-n_1}\lrp{\frac{n_1-Np_*}{p_*(1-p_*)}}^2 \\
    \leq \eta_3\frac{N^2}{(p_*(1-p_*))^2}\sum_{n_1\leq n_1^{(3)}} \binom{N}{n_1}p_*^{n_1}(1-p_*)^{N-n_1}
\end{align*}
This final sum can now be bounded from above using the Chernoff bound \cite{arratia_tutorial_1989},
\begin{align*}
\sum_{n_1\leq n_1^{(3)}}\binom{N}{n_1} p_*^{n_1}(1-p_*)^{N-n_1}\leq \exp\lrp{-ND\lrp{\frac{n_1^{(3)}}{N}\Big\| p_*}},
\end{align*}
where $D\lrp{\frac{n_1^{(3)}}{N}\|p_*}$ is the Kullback-Leibler (KL) divergence for a binomial distribution. For any $\epsilon_3>0$, the KL divergence is positive when $\frac{n_1^{(3)}}{N}<p_*-\epsilon_3$. For a positive KL divergence, the partial sum is exponentially suppressed in the limit of large $N$. In other words, $\forall \ \epsilon_3>0$, $\exists \ N_1^{(3)}$ such that $N>N_1^{(3)}$ implies
\begin{align*}
\eta_3\frac{N^2}{p_*(1-p_*)}\sum_{n_1\leq n_1^{(3)}} p_*^{n_1}(1-p_*)^{N-n_1}<\epsilon_3,
\end{align*}
and so
\begin{align*}
\tx{RHS}< \eta_3\sum_{n_1> n_1^{(3)}} \frac{\lrp{p_*^{n_1}(1-p_*)^{N-n_1}\frac{n_1-Np_*}{p_*(1-p_*)}}^2}{p_*^{n_1}(1-p_*)^{N-n_1}}+\epsilon_3.
\end{align*}
Taking $\epsilon_3<1-\eta_3$ proves the inequality of \cref{eq:third_bound}.

\end{proof}

\section{Estimate for the variance of the FI estimator}\label{section:estimate_variance_estimator}
We calculate the variance of \cref{eq:estimator_FI}. 
\begin{align}
    \tx{Var}\lrb{\hat{\mathcal{I}}(\omega)}&=\tx{Var}\lrb{\frac{1}{S}\sum_{s=1}^S\lrp{\frac{\del_{\omega}P(x_1^{(s)},\cdots,x_N^{(s)})}{P(x_1^{(s)},\cdots,x_N^{(s)})}}^2} \nonumber \\
    &=\frac{1}{S^2}\sum_{s=1}^S\tx{Var}\lrb{\lrp{\frac{\del_{\omega}P(x_1^{(s)},\cdots,x_N^{(s)})}{P(x_1^{(s)},\cdots,x_N^{(s)})}}^2}.\label{eq:variance_proof_1}
\end{align}
There are no covariance terms in \cref{eq:variance_proof_1} since all of our samples are uncorrelated, as they are produced via direct sampling of the probability distribution. Using the definition of the variance, we obtain
\begin{align}
    \tx{Var}\lrb{\hat{\mathcal{I}}(\omega)}=\frac{1}{S^2}\sum_{s=1}^S\tx{Var}\lrb{\lrp{\frac{\del_{\omega}P(x_1^{(s)},\cdots,x_N^{(s)})}{P(x_1^{(s)},\cdots,x_N^{(s)})}}^2}\nonumber \\
    =\frac{1}{S^2}\sum_{s=1}^S\left(\mathbb{E}\lrb{\lrp{\frac{\del_{\omega}P(x_1^{(s)},\cdots,x_N^{(s)})}{P(x_1^{(s)},\cdots,x_N^{(s)})}}^4}\right.\nonumber \\
    \quad \quad \left.-\mathbb{E}\lrb{\lrp{\frac{\del_{\omega}P(x_1^{(s)},\cdots,x_N^{(s)})}{P(x_1^{(s)},\cdots,x_N^{(s)})}}^2}^2 \right). \nonumber
\end{align}
The expectation value is independent of the sample. Further, we recognize the definition of the FI in the second expectation value. The equation then becomes
\begin{align*}
    \tx{Var}\lrb{\hat{\mathcal{I}}(\omega)}&=\frac{1}{S}\mathbb{E}\lrb{\lrp{\frac{\del_{\omega}P(x_1,\cdots,x_N)}{P(x_1,\cdots,x_N)}}^4} - \frac{\mathcal{I}(\omega)^2}{S}.
\end{align*}
We can upper bound the variance,
\begin{align}
    \tx{Var}\lrb{\hat{\mathcal{I}}(\omega)}\leq \frac{1}{S}\mathbb{E}\lrb{\lrp{\frac{\del_{\omega}P(x_1,\cdots,x_N}{P(x_1,\cdots,x_N}}^4}.
\end{align}
The right side of this equation can be estimated with
\begin{align}\label{eq:variance_estimator}
    \frac{1}{S^2}\sum_{s=1}^S \lrp{\frac{\del_{\omega}P(x_1,\cdots,x_N)}{P(x_1,\cdots,x_N}}^4.
\end{align}
We therefore take \cref{eq:variance_estimator} to be a biased estimator of our variance, allowing us to draw error bars in \cref{fig:heisenberg}. Although the estimator is biased, it converges to a value that is higher than the true value. The small error bars on \cref{fig:heisenberg} are therefore a good indicator of the convergence of the FI estimator. This method remains a heuristic, as although \cref{eq:variance_estimator} converges for high-enough $S$, it is unclear how different values of $N$ affect the convergence and whether the 25 000 samples are sufficient.

\section{List of Acronyms}\label{app:acronyms}

{\renewcommand\labelitemi{}
\begin{itemize}
\item\textbf{QFI}\quad Quantum Fisher information
\item\textbf{HL}\quad Heisenberg limit
\item\textbf{SQL}\quad Standard quantum limit
\item\textbf{QEC}\quad Quantum error correction
\item\textbf{HMM}\quad Hidden Markov model
\item\textbf{KL}\quad Knill-Laflamme
\item\textbf{CPTP}\quad Completely positive trace-preserving
\item\textbf{HNLS}\quad Hamiltonian not in Lindblad span
\item\textbf{HNES}\quad Hamiltonian not in extended span
\item\textbf{HNELS}\quad Hamiltonian not in extended Lindblad span
\item\textbf{FI}\quad Fisher information, or classical Fisher information
\end{itemize}
}

\end{document}